\documentclass[preprint,12pt]{elsarticle}

\usepackage{amssymb}
\usepackage{amsmath}
\usepackage[a4paper]{geometry}
\usepackage{graphicx}
\usepackage{amsthm}
\usepackage{delarray}
\usepackage{enumitem}
\usepackage{tikz}
\usepackage{caption}
\usepackage{hyperref}
\usepackage{stmaryrd}
\usepackage{ifthen}
\usepackage{calc}
\usepackage{float}
\usepackage[ruled,vlined,linesnumbered]{algorithm2e}
\usepackage{etoolbox}

\newtheorem{theorem}{\textbf{Theorem}}
\newtheorem{lemma}{\textbf{Lemma}}
\newtheorem{proposition}{\textbf{Proposition}}
\newtheorem{corollary}{\textbf{Corollary}}
\newtheorem{remark}{\textbf{Remark}}

\newtheorem{definition}{\textbf{Definition}}
\newtheorem{notation}{\textbf{Notation}}
\newtheorem{remark*}{Remark}
\newtheorem{example}{\textbf{Example}}
\newtheorem{open}[theorem]{\textbf{Open problem}}
\newtheorem{conjecture}{\textbf{Conjecture}}

\journal{Theoretical Computer Science}

\newcommand{\N}{\mathbb{N}}
\newcommand{\Z}{\mathbb{Z}}

\newcommand{\B}{\mathbb{B}}
\newcommand{\1}{\mathtt{1}}
\newcommand{\0}{\mathtt{0}}
\newcommand{\bool}{\{\0,\1\}}
\renewcommand{\O}{\mathcal{O}}
\newcommand{\Poly}{{\mathsf{P}}}
\newcommand{\PPoly}{{\oplus\mathsf{P}}}
\newcommand{\ModkPoly}[1]{{\mathsf{Mod_{#1}P}}}
\newcommand{\ModPoly}{{\mathsf{ModP}}}
\newcommand{\SPoly}{{\#\mathsf{P}}}

\newcommand{\NP}{{\mathsf{NP}}}
\newcommand{\coNP}{{\mathsf{coNP}}}

\newcommand{\PSPACE}{{\mathsf{PSPACE}}}
\newcommand{\coPSPACE}{{\mathsf{coPSPACE}}}

\newcommand{\FPoPSPACE}{{\mathsf{FP}}^{\mathsf{PSPACE}}}

\newcommand{\entiers}[2][2]{\ifthenelse{\equal{#1}{2}}
{\llbracket #2\rrbracket}
{\ifthenelse{\equal{#1}{0}}{
\ifthenelse{\equal{11}{\the\catcode`#2}}
{\{0,\ldots,#2 -1\}}
{\newcounter{finaln}
\setcounter{finaln}{#2 -1} 
\{0,\ldots,\thefinaln \}}}
{\{1,\ldots,#2\}}}}

\renewcommand{\int}[1]{[{#1}]}
\newcommand{\mmjblock}[2]{{#1}_{(#2)}}
\newcommand{\mmjoblock}[2]{{#1}_{\{#2\}}}
\newcommand{\mmjoblocklimit}[2]{{#1}^{\Omega}_{\{#2\}}}

\newcommand{\Wl}{(W_i)_{i\in\entiers{\ell}}}
\newcommand{\Sk}{\{S_k\}_{k\in\entiers{s}}}

\newcommand{\ispe}{\hat{\imath}}
\newcommand{\jspe}{\hat{\jmath}}
\newcommand{\kspe}{\hat{k}}

\newcommand{\BPn}{\mathsf{BP}_n}

\newcommand{\BSn}{\mathsf{BS}_n}
\newcommand{\BS}[1]{\mathsf{BS}_{#1}}

\newcommand{\lcm}{\text{lcm}}
\renewcommand{\gcd}{\text{gcd}}

\newcommand{\decisionpb}[4]{\fbox{\parbox{.98\textwidth}{
  {#1} ({\bf #2})\\
  {Input:} #3\\
  {Question:} #4
}}}

\graphicspath{{figures/}}%helpful if your graphic files are in another directory

\begin{document}

\begin{frontmatter}

\title{Foundations of block-parallel automata networks}

\author[lis,univpub]{K{\'e}vin Perrot} %% Author name
\author[lis,univpub]{Sylvain Sen{\'e}}
\author[lis]{L{\'e}ah Tapin}

\affiliation[lis]{organization={Aix-Marseille Univ., CNRS, LIS},%Department and Organization
            %addressline={}, 
            city={Marseille},
            %postcode={}, 
            %state={},
            country={France}}

\affiliation[univpub]{organization={Université publique},%Department and Organization
            %addressline={}, 
            city={Marseille},
            %postcode={}, 
            %state={},
            country={France}}

%% Abstract
\begin{abstract}
  We settle the theoretical ground for the study of automata networks under block-parallel update schedules,
  which are somehow dual to the block-sequential ones,
  but allow for repetitions of automaton updates.
  This gain in expressivity brings new challenges,
  and we analyse natural equivalence classes of update schedules:
  those leading to the same dynamics, and to the same limit dynamics, for any automata network.
  Countings and enumeration algorithms are provided, for their numerical study.
  We also prove computational complexity bounds for many classical problems,
  involving fixed points, limit cycles, the recognition of subdynamics, reachability, \emph{etc}.
  The $\PSPACE$-completeness of computing the image of a single configuration lifts
  the complexity of most problems, but the landscape keeps some relief,
  in particular for reversible computations.
\end{abstract}



%% Keywords
\begin{keyword}
  Automata networks\sep
  block-parallel\sep
  update modes\sep
  equivalence classes\sep
  counting\sep
  enumeration\sep
  computational\sep
  complexity.



\end{keyword}

\end{frontmatter}

\section{Introduction}

Since the seminal work of McCulloch and Pitts on neural networks~\cite{J-McCulloch1943},
distributed models of computation, where individual entities collectively
perform a global computation through local interactions,
received a great amount of attention.
Automata networks are such a model, which have successfully been employed
for the modelling of gene regulation
mechanisms~\cite{J-Kauffman1969,J-Thomas1973}.
In particular, the biological interpretation of the limit dynamics of automata networks
matches experimental results~\cite{J-Mendoza1998,C-Akutsu1999,J-Giacomantonio2010,J-Wooten2019}.

One can readily observe that distributed models of computation
are highly sensitive to variations in the update schedule among its entities.
Regarding Boolean automata networks,
despite the fact that fixed points obtained under the parallel update mode
are also fixed points for any other update schedule~\cite{B-Goles1990},
specific update modes may generate additional fixed points~\cite{J-Demongeot2020,BC-Pauleve2022}.
Limit cycles are also known to greatly depend on the update
mode~\cite{J-Demongeot2008,J-Goles2008,J-Aracena2009,C-Goles2010,J-Aracena2013}.

In this work, we propose to address a new family of update schedules,
namely the block-parallel update modes,
which are motivated by the discovery of the importance of chromatin dynamics
in regulatory networks~\cite{J-Demongeot2020}.
Indeed, it fits our current understanding of the temporality of mRNA transcriptional
machinery~\cite{J-Hansen1992,J-Benecke2006,J-Hubner2010,J-Fierz2019}.
Block-parallel update modes permit local update repetitions,
opening new doors towards the phenomenological modelling in systems biology.
Their novel features challenge intuitions erected upon
classical block-sequential ones, and 
our objective is to build solid theoretical foundations for their study.

On the one hand, we study the combinatorics of block-parallel modes,
and provide countings and enumeration algorithms
for various meaningful equivalence relations.
This is the grounding necessary to perform efficient numerical simulations.
On the other hand, we analyse the computational cost of standard
decision problems one may be willing to answer on the dynamics
of Boolean automata networks with these new update schedules
(related ot fixed points, limit cycles, reachability, \emph{etc}).
While block-parallel update modes seem to be more expressive,
in particular regarding the limit dynamics,
this gain of expressivity comes at a high cost in terms of simulation.
Indeed, most problems traditionally $\NP$-complete become
$\PSPACE$-complete, however there are notable exceptions.\\
Note: This article is a long version dedicated to put together basic results on 
block-parallel update modes in the framework of automata networks which have been published respectively in the proceedings of SOFSEM 2024 and SAND 2024~\cite{C-Perrot2024a,C-Perrot2024b}.

%%%

In Section~\ref{s:def}, we present the main definitions and notations.
This paper is divided in two main sections, the first one focused on the counting
and enumeration of block-parallel update modes and the second on complexity issues.
The first part, Section~\ref{s:contributions}, is split into five subsections.
Subsection~\ref{s:BSinterBP} serves as a sort of introduction, by dealing with the
intersection between block-sequential and block-parallel update modes.
The three following subsections each present a subset of block-parallel update modes,
by giving a formula for counting the elements of this subset, and an algorithm to enumerate
them.
Subsection~\ref{s:BP} deals with the whole set of block-parallel update modes,
Subsection~\ref{s:BP0} with the updates schedules up to dynamical equality, and 
Subsection~\ref{s:BPstar} with the update schedules up to dynamical isomorphism
on the limit set, which is the one we focused more on.
The last subsection of this part presents the results from our implementations of
the aforementioned algorithms.
Section~\ref{s:complexity} exposes our results regarding complexity problems involving
block-parallel update schedules.
In Subsection~\ref{ss:image},
we first characterize classical problems on computing images, preimages, fixed points
and limit cycles : they all jump from $\NP$ (under block-sequential update modes) to 
$\PSPACE$ (under block-parallel update modes).
In Subsection~\ref{ss:reach},
we then prove a general bound on the recognition of functional subdynamics.
Regarding global properties, recognizing bijective dynamics remains $\coNP$-complete,
and recognizing constant dynamics becomes $\PSPACE$-complete.
The case of identity recognition is much subtler, and we provide three incomparable bounds:
a trivial $\coNP$-hardness one, a tough $\ModPoly$-hardness, and a $\FPoPSPACE$-completeness
result derived from the recent literature.
Finally, we summarize our results and expose some perspectives in Section~\ref{s:conclusion}.

\section{Definitions and state of the art}
\label{s:def}

We denote the set of integers by $\entiers{n} = \entiers[0]{n}$,
the Booleans by $\B = \bool$,
the $i$-th component of a vector $x \in \B^n$ by $x_i \in \B$,
and the restriction of $x$ to domain $I\subset\entiers{n}$ by $x_I \in \B^{|I|}$.
Let $e_i$ be the $i$-th base vector,
and $\forall x, y \in \B^n$, let $x+y$ denote the bitwise addition modulo two.
Let $\sigma^i$ denote the circular-shift of order $i\in\Z$ on sequences
(shifting the element at position $0$ towards position $i$).
For two graphs $G = (V(G),A(G))$ and $H = (V(H),A(H))$,
we denote by $G \sim H$ when they are isomorphic,
i.e.,~when there is a bijection
$\pi : V(G) \to V(H)$ such that $(x,y) \in A(G) \iff (\pi(x), \pi(y)) \in A(H)$.
We denote by $G\sqsubset H$ when $G$ is a subgraph of $H$,
i.e.,~when $G'$ such that $G'\sim G$ can be obtained from $H$ by vertex and arc deletions.

\noindent
\textbf{Boolean automata network}\quad
A \emph{Boolean automata network} (BAN) is a discrete dynamical system on $\B^n$.
A configuration $x\in\B^n$ associates to each of the $n$ automata among $\entiers{n}$ a Boolean state among $\B$.
The individual dynamics of a each automaton $i\in\entiers{n}$
is described by a local function $f_i:\B^n\to\B$ giving its new state according to the current configura\-tion.
To get a dynamics, one needs to settle the order in which the automata update their state
by application of their local function.
That is, an \emph{update schedule} must be given.
The most basic is the parallel update schedule,
where all automata update their state synchronously at each step,
formally as $f:\B^n\to\B^n$ defined by $\forall x\in\B^n:f(x)=(f_0(x),f_1(x),\dots,f_{n-1}(x))$.
In this work, we concentrate on the block-parallel update schedule,
motivated by the biological context of gene regulatory networks,
where each automaton is a gene and the dynamics give clues on cell phenotypes.
Not all automata will be update simultaneously as in the parallel update mode.
They will instead be grouped by subsets.
For simplicity in defining the local functions of a BAN,
we extend the $f_i:\B^n\to\B$ notation to
subsets $I\subseteq\entiers{n}$ as $f_I:\B^n\to\B^{|I|}$.
We also denote $\mmjblock{f}{I}:\B^n\to\B^n$ the update of automata from subset $I$, defined as:
\[
  \forall i \in \entiers{n}:
  \mmjblock{f}{I}(x)_i =
  \begin{cases}
    f_i(x) & \text{if } i \in I\\
    x_i & \text{otherwise.}
  \end{cases}
\]

\noindent
\textbf{Block-sequential update schedule}\quad
A \emph{block-sequential} update schedule is an \emph{ordered partition} of $\entiers{n}$,
given as a sequence of subsets $\Wl$ where $W_i \subseteq \entiers{n}$
is a \emph{block}.
The automata within a block are updated simultaneously,
and the blocks are updated sequentially.
During one iteration (\emph{step}) of the network,
the state of each automaton is updated exactly once.
The update of each block is called a \emph{substep}.
This update mode received great attention on many aspects.
The concept of the \emph{update digraph} is introduced in~\cite{J-Aracena2011}
and characterized in~\cite{J-Aracena2013b}
to capture equivalence classes of block-sequential update schedules (leading to the same dynamics).
Conversions between block-sequential and parallel update schedules are investigated
in~\cite{C-Perrotin2023} (how to parallelize a block-sequential update schedule),
\cite{C-Goles2010} (the preservation of cycles throughout the parallelization process),
and~\cite{C-Bridoux2017} (the cost of sequentialization of a parallel update schedule).

\medskip\noindent
\textbf{Block-parallel update schedule}\quad
A \emph{block-parallel} update schedule is a \emph{partitioned order} of $\entiers{n}$,
given as a set of subsets $\mu=\Sk$ where $S_k = (i^k_0, \dots, i^k_{n_{k}-1})$
is a sequence of $n_k>0$ elements of $\entiers{n}$ for all $k \in \entiers{s}$, called an \emph{o-block}
(shortcut for \emph{ordered-block}).
Each automaton appears in exactly one o-block.
It follows an idea dual to the block-sequential update mode:
the automata within an o-block are updated sequentially,
and the o-blocks are updated simultaneously.
The o-block sequences are taken circularly at each substep,
until we reach the end of each o-block simultaneously
(which happens after the least common multiple ($\lcm$) of their sizes).
The set of block-parallel update modes of size $n$ is denoted $\BPn$.
Formally, the update of $f$ under $\mu\in\BPn$ is given by
$\mmjoblock{f}{\mu} : \B^n \to \B^n$ defined, with $\ell=\lcm(n_1,\dots,n_s)$, as
%\[
  $\mmjoblock{f}{\mu}(x) = 
  \mmjblock{f}{W_{\ell-1}} \circ
  \dots \circ
  \mmjblock{f}{W_1} \circ
  \mmjblock {f}{W_0}(x)$,
%\]
where for all $i \in \entiers{\ell}$ we define 
$W_i = \{i^k_{i \mod n_k} \mid k \in \int{s}\}$.
In order to compute the set of automata updated at each substep,
it is possible to convert a block-parallel update schedule into
a sequence of blocks of length $\ell$
(which is usually not a block-sequential update schedule,
because repetitions of automaton update may appear).
We defined this map as $\varphi$:
\[
  \varphi(\Sk) = 
  \Wl \text{ with } W_i = 
  \{i^k_{i \mod n_k} \mid k \in \int{s}\}\text{.}
\]
An example is given on Figure~\ref{fig:example}.
The parallel update schedule corresponds to the
block-parallel update schedule $\mu_\texttt{par}=\{(i) \mid i \in \entiers{n}\} \in \BPn$,
with $\varphi(\mu_\texttt{par})=(\entiers{n}$),
i.e., a single block containing all automata is updated at each step (there is only one substep).

\begin{figure}
  \centering
  \includegraphics[scale=.97]{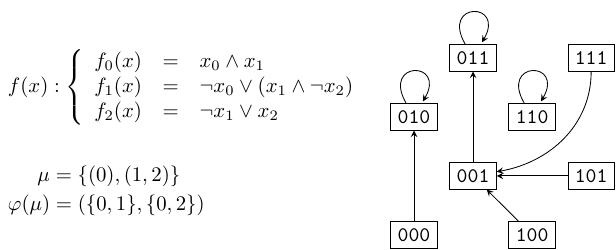}
  \caption{
    Example of an automata network of size $n=3$ with a block-parallel update mode $\mu\in\BPn$.
    Local functions (upper left), conversion of $\mu$ to a sequence of blocks (lower left),
    and dynamics of $\mmjoblock{f}{\mu}$ on configuration space $\B^3$ (right).
    One step is composed of two substeps: the first substep updates the block $\{0,1\}$,
    the second substep updates the block $\{0,2\}$.
    As an example, in computing the image of configuration $\1\1\1$,
    the first substep (update of automata $0$ and $1$) gives $\1\0\1$,
    and the second substep (update of automata $0$ and $2$) gives $\0\0\1$.
  }
  \label{fig:example}
\end{figure}

Block-parallel update schedules have been introduced in~\cite{J-Demongeot2020},
motivated by applications to gene regulatory networks, and 
their ability to generate new stable configurations
(compared to block-sequential update schedules).
%A first theoretical study has been conducted in~\cite{C-Perrot2024},
%providing counting formulas and enumeration algorithms,
%subject to equivalence relations on the produced dynamics.

\medskip\noindent
\textbf{Fixed point and limit cycle}\quad
A BAN $f$ of size $n$ under block-parallel update schedule $\mu\in\BPn$
defines a deterministic discrete dynamical system $\mmjoblock{f}{\mu}$ on configuration space $\B^n$.
Since the space is finite, the orbit of any configuration
is ultimately periodic.
For $p\geq 1$, a sequence of configurations $x^0,\dots,x^{p-1}$ is a \emph{limit cycle} of length $p$
when $\forall i\in\entiers{p}:\mmjoblock{f}{\mu}(x^i)=x^{i+1\mod p}$.
For $p=1$ we call $x\in\B^n$ such that $\mmjoblock{f}{\mu}(x)=x$ a \emph{fixed point}.

\medskip\noindent
\textbf{Complexity}\quad
To be given as input to a decision problem, a BAN is encoded as a tuple of
$n$ Boolean circuits, one for each local function $f_i:\B^N\to\B$ for $i\in\entiers{n}$.
This encoding can be seen as Boolean formulas for each automaton,
and easily implements high-level descriptions with if-then-else statements
(used intensively in our constructions).

The computational complexity of finite discrete dynamical systems has been explored
on the related models of finite cellular automata~\cite{J-Sutner1995}
and reaction networks~\cite{J-Dennunzio2019}.
Regarding automata networks, fixed points received early attention in~\cite{J-Alon1985}
and~\cite{J-Floreen1989}, with existence problems complete for $\NP$.
Because of the fixed point invariance for block-sequential update schedules~\cite{B-Robert1986},
the focus switched to limit cycles~\cite{J-Aracena2013,C-Bridoux2021},
with problems reaching the second level of the polynomial hierarchy.
The interplay of different update schedules has been investigated in~\cite{J-Aracena2013}.
Finaly, let us mention the general complexity lower bounds,
established for any first-order question on the dynamics,
under the parallel update schedule~\cite{C-Gamard2021}.

%%%%%%%%%%%%%%%%%%%%%%%%%%%% ENUM

\section{Counting and enumerating block-parallel update modes}
\label{s:contributions}

For the rest of this section,
let $p(n)$ denote the number of integer partitions of $n$
(multisets of integers summing to $n$),
let $d(i)$ be the maximal part size in the $i$-th partition of $n$,
let $m(i,j)$ be the multiplicity of the part of size $j$ in the $i$-th partition 
of $n$.
As an example, let $n = 31$ and assume the $i$-th partition is 
$(2,2,3,3,3,3,5,5,5)$, we have $d(i) = 5$ and $m(i,1) = 0$, $m(i,2) = 2$, 
$m(i,3) = 4$, $m(i,4) = 0$, $m(i,5) = 3$.
A partition will be the support of a partitioned order, where each part is an 
o-block.
In our example, we can have:
\[
  \begin{array}{l}
    \{(0,1),(2,3),(4,5,6),(7,8,9),(10,11,12),(13,14,15),\\
    (16,17,18,19,20),(21,22,23,24,25),(26,27,28,29,30)\}\text{,}
  \end{array}
\]
and we picture it as the following \emph{matrix-representation}:
\[
  \begin{pmatrix}
    0&1\\
    2&3
  \end{pmatrix}
  \begin{pmatrix}
    4&5&6\\
    7&8&9\\
    10&11&12\\
    13&14&15
  \end{pmatrix}
  \begin{pmatrix}
    16&17&18&19&20\\
    21&22&23&24&25\\
    26&27&28&29&30
  \end{pmatrix}
  \text{.}
\]
We call \emph{matrices} the elements of size $j \cdot m(i,j)$ and denote them 
$M_1, \ldots, M_{d(i)}$, where $M_j$ has $m(i,j)$ \emph{rows} and $j$ 
\emph{columns} ($M_j$ is empty when $m(i,j) = 0$).
The partition defines the matrices' dimensions, and each row is an o-block.\medskip

For the comparison, the block-sequential update modes (ordered partitions of 
$\entiers{n}$) are given by the ordered Bell numbers, sequence 
A000670 of OEIS~\cite{oeisA000670,ns11}.
A closed formula for it is:
\[
  |\BS{n}|=
	  \sum_{i = 1}^{p(n)}
  		\frac{n!}{\prod_{j = 1}^{d(i)} 
  			(j!)^{m(i,j)}}
  	\cdot
	  \frac{\left( \sum_{j = 1}^{d(i)} m(i,j) \right)!}
		  {\prod_{j = 1}^{d(i)} m(i,j)!}\text{.}
\]
Intuitively, an ordered partition of $n$ %(also called a combination of $n$)
gives a support to construct a block-sequential update mode:
place the elements of $\entiers{n}$ up to permutation within the blocks.
This is the left fraction: $n!$ divided by $j!$ for each block of size $j$,
taking into account multiplicities.
The right fraction corrects the count because we sum on $p(n)$ the (unordered) 
partitions of $n$:
each partition of $n$ can give rise to different %combinations of $n$,
ordered partitions of $n$, 
by ordering all blocks (numerator, where the sum of multiplicities is the number 
of blocks) up to permutation within blocks of the same size which have no effect 
(denominator).
The first ten terms are ($n = 1$ onward):
\[
	1, 3, 13, 75, 541, 4683, 47293, 545835, 7087261, 102247563\text{.}
\]

%A recursive formula is given by:
%\[
%  |\BS{n}|=\sum_{i=0}^{n-1}\binom{n}{i}\cdot|\BS{i}| \quad\text{ width }\quad |\BS{1}|=1.
%\]
%Intuitively, choose $0$ to $n-1$ automata among $\entiers{n}$,
%count ordered partitions recursively, and place all the remaining automata
%is a fresh last block. This produce all ordered partitions of $\entiers{n}$
%(\emph{e.g.}~the case $i=0$ corresponds to the block-sequential mode $(\entiers{n})$).

%%%%%%%%%%%%%%%%
\subsection{Intersection of block-sequential and block-parallel modes}
\label{s:BSinterBP}

In order to be able to compare block-sequential with block-parallel update
modes, both of them will be written here under their sequence of blocks form
(the classical form for block-sequential update modes and the rewritten form for 
block-parallel modes).

First, we know that $\varphi(\textsf{BP}_n) \cap \textsf{BS}_n$ is not empty, 
since it contains at least 
\[
  \mu_\textsf{par}
  = (\entiers{n})
  = \varphi(\{(0), (1), \ldots, (n-1)\})\text{.}
\]
However, neither $\BSn \subseteq \varphi(\BPn)$ nor $\varphi(\BPn) \subseteq
\BSn$ are true.
Indeed, $\mu_s = (\{0, 1\}, \{2\}) \in \mathsf{BS}_3$ but $\mu_s \notin
\varphi(\mathsf{BP}_3)$ since a block-parallel cannot have blocks of different 
sizes in its sequential form. 
Symmetrically, $\mu_p = \varphi(\{(1, 2), (0)\}) = 
(\{0,1\}, \{0, 2\}) \in \mathsf{BP}_3$ but $\mu_p \notin \mathsf{BS}_3$ since
automaton $0$ is updated twice.
Despite this, we can precisely define the intersection 
$\BSn \cap \varphi(\BPn)$.

\begin{lemma}
	\label{lemma:BPequivBS}
  Let $\mu$ be an update mode written as a sequence of blocks of elements in 
  $\entiers{n}$.  
  Then $\mu \in (\BSn \cap \varphi(\BPn))$ if and only if $\mu$ is 
  an ordered partition and all of $\mu$'s blocks are of the same size.
\end{lemma}

\begin{proof}
  Let $n \in \N$.
  
  $(\Longrightarrow)$ Let $\mu \in (\BSn \cap \varphi(\BPn))$. 
  Since $\mu \in \BSn$, $\mu$ is an ordered partition.
  Furthermore, $\mu \in \varphi(\BPn)$ so all the $\mu$'s blocks are of the same 
  size.

  $(\Longleftarrow)$ Let $\mu = \Wl$ be an ordered partition of $\entiers{n}$ with all its blocks 
  having the same size, denoted by $s$.
  Since $\mu$ is an ordered partition, $\mu \in \BSn$.
  For each $\ell \in \entiers{p}$, we can number arbitrarily the elements of $W_\ell$ 
  from $0$ to $s-1$ as $W_\ell=\{W_\ell^0,\dots,W_\ell^{s-1}\}$.
  Now, let us define the set of sequences $\Sk$ the following way: 
  $\forall k \in \entiers{s}, S_k = 
  \{ W_\ell^k \mid \ell\in\entiers{p} \}$.
  %\{i \in \entiers{n} \mid \exists \ell \in  \entiers{p}, (W_\ell)_k = i\}$.
  It is a partitioned order such that $\varphi(\Sk) = \mu$, which means
  that $\mu \in \varphi(\BPn)$.
\end{proof}

\begin{corollary}
	\label{cor:BPequivBS}
	If $\mu \in \BPn$ and is composed of $s$ o-blocks of size $p$, then 
	$\varphi(\mu) \in \BSn$ and is composed of $p$ blocks of size $s$.	
\end{corollary}

As a consequence of Lemma~\ref{lemma:BPequivBS} and Corollary~
\ref{cor:BPequivBS}, given $n \in \N$, the set $\mathrm{SEQ}_n$ of sequential 
update modes such that every automaton is updated exactly once by step and only 
one automaton is updated by substep, is a subset of $(\BSn \cap \varphi(\BPn))$.

Moreover, we can state the following proposition which counts the number of 
sequences of blocks which belongs to both $\BSn$ and $\varphi(\BPn)$.

\begin{proposition}
	\label{prop:countBSinterBP}
	Given $n \in \N$, we have:
	\[
		|\BSn\cap \varphi(\BPn)| =
			\sum_{d | n}\frac{n!}{(\frac{n}{d}!)^d}\text{.}
	\]
\end{proposition}

\begin{proof}
	The proof derives directly from the sequence A061095 of 
	OEIS~\cite{oeisA061095}, which counts the \emph{number of ways of dividing $n$ 
	labeled items into labeled boxes with an equal number of items in each box}.
	In our context, the ``items'' are the automata, and the ``labeled boxes'' are 
	the blocks of the ordered partitions.
\end{proof}

%%%%%%%%%%%%%%%%
\subsection{Partitioned orders}
\label{s:BP}

A block-parallel update mode is given as a partitioned order, \emph{i.e.}~an 
(unordered) set of (ordered) sequences.
This concept is recorded as sequence A000262 of OEIS~\cite{oeisA000262}, 
described as the \emph{number of ``sets of lists''}.
A nice closed formula for it is:
\[
  |\BPn| = 
  	\sum_{i = 1}^{p(n)} 
  		\frac{n!}{\prod_{j=1}^{d(i)} 
  			m(i,j)!}\text{.}
\]
Intuitively, for each partition, fill all the matrices
($n!$ ways to place the elements of $\entiers{n}$)
up to permutation of the 
rows within each matrix
(matrix $M_j$ has $m(i,j)$ rows).
Another closed formula is presented in Proposition~\ref{prop:cardBPn}.
This formula is particularly useful to generate all the block-parallel update 
modes.

\begin{proposition}
	\label{prop:cardBPn}
  For any $n \geq 1$ we have:
  \[
  |\BPn| =
	  \sum_{i = 1}^{p(n)} 
	  	\prod_{j = 1}^{d(i)}
	  		\binom{n - \sum_{k = 1}^{j - 1} k \cdot m(i,k)}{j \cdot m(i,j)}
  				\cdot \frac{(j \cdot m(i,j))!}{m(i,j)!}\text{.}
  \]
\end{proposition}

\begin{proof}
	Each partition is a support to generate different partitioned orders (sum on 
	$i$), by considering all the combinations, for each matrix (product on $j$),
	of the ways to choose the $j \cdot m(i,j)$ elements of $\entiers{n}$
	it contains (binomial coefficient, chosen among the remaining elements),
	and all the ways to order them up to permutation of the rows (ratio of 
	factorials).
	Observe that developing the binomial coefficients with
	${\binom{x}{y}} = \frac{x!}{y! \cdot (x-y)!}$ gives
	\[
  	\prod_{j = 1}^{d(i)}
  		\binom{n - \sum_{k}}{j \cdot m(i,j)} \cdot(j \cdot m(i,j))!
		= \prod_{j = 1}^{d(i)}
			\frac{(n - \sum_{k})!}{(n - \sum_{k} - j \cdot m(i,j))!}
		= \frac{n!}{0!}
		= n!\text{,}
	\]
	where $\sum_{k}$ is a shorthand for $\sum_{k = 1}^{j - 1} k \cdot m(i,k)$, 
	which leads to retrieve the OEIS formula.
\end{proof}

The first ten terms are ($n = 1$ onward):
\[
	1, 3, 13, 73, 501, 4051, 37633, 394353, 4596553, 58941091\text{.}
\]

The formula from Proposition~\ref{prop:cardBPn} can give us an algorithm to
enumerate the partitioned orders of size $n$.
For each partition $i$ of $n$, it enumerates all the block-parallel update modes as a list of 
matrices as presented in the introduction of this section, with the matrices
being filled one-by-one (the product on $j$ in the formula).
For each value of $j$, we choose a set of $j \cdot m(i,j)$ elements of $\entiers{n}$
(that haven't been placed in a previous matrix) to put in matrix $M_j$.
We then enumerate all the different ways to agence these numbers,
up to permutation of the rows.

%%%%%%%%%%%%%%%%
\subsection{Partitioned orders up to dynamical equality}
\label{s:BP0}

As for block-sequential update modes, given an AN $f$ and two block-parallel 
update modes $\mu$ and $\mu'$, the dynamics of $f$ under $\mu$ can be the same 
as that of $f$ under $\mu'$. 
To go further, in the framework of block-parallel update modes,
there exist pairs of update modes $\mu,\mu'$ such that for any AN $f$,
the dynamics $\mmjoblock{f}{\mu}$ is the exact same as $\mmjoblock{f}{\mu'}$.
%In the context of BAN update modes, we can observe that among $\BPn$,
%some update modes will always give the exact same dynamics.
As a consequence, in order to perform exhaustive searches among
the possible dynamics, it is not necessary to generate all of them.
We formalize this with the following equivalence relation.

\begin{definition}
  \label{def:dyn_eq}
  For $\mu,\mu'\in\BPn$, we denote $\mu\equiv_0\mu'$
  when $\varphi(\mu)=\varphi(\mu')$.
\end{definition}

The following Lemma shows that this equivalence relation is necessary and 
sufficient in the general case of ANs of size $n$.
%The following Lemma shows that this equivalence relation
%is necessary and sufficient in the general case of BANs of size $n$.

\begin{lemma}
  For any $\mu, \mu' \in \BPn$, we have $\mu \equiv_0 \mu' \iff 
  \forall f: X \to X, \mmjoblock{f}{\mu} = \mmjoblock{f}{\mu'}$.
%  For any $\mu,\mu'\in\BPn$, we have $\mu\equiv_0\mu' \iff \forall
%f:\B^n\to\B^n, \mmjoblock{f}{\mu}=\mmjoblock{f}{\mu'}$.
\end{lemma}

\begin{proof}
	Let $\mu$ and $\mu'$ be two block-parallel update modes of $\BPn$.\medskip
	
  \noindent $(\Longrightarrow)$ Let us consider that $\mu\equiv_0\mu'$, 
  and let $f : X \to X$ be an AN. 
%  and let $f : \B^n \to \B^n$ be a BAN. 
  Then, we have $\mmjoblock{f}{\mu} = \mmjblock{f}{\varphi(\mu)} =
  \mmjblock{f}{\varphi(\mu')} = \mmjoblock{f}{\mu'}$.\medskip

  \noindent $(\Longleftarrow)$ Let us consider that 
  $\forall f : X \to X, \mmjoblock{f}{\mu} = \mmjoblock{f}{\mu'}$. 
%  $\forall f : \B^n \to \B^n, \mmjoblock{f}{\mu}=\mmjoblock{f}{\mu'}$. 
  Let us assume for the sake of contradiction that $\varphi(\mu) \neq
  \varphi(\mu')$. 
  For ease of reading, we will denote as $t_{\mu, i}$ the substep at which
  automaton $i$ is updated for the first time with update mode $\mu$. 
  Then, there is a pair of automata $(i,j)$ such that 
  $t_{\mu, i} \leq t_{\mu, j}$, but $t_{\mu', i} > t_{\mu',j}$. 
  Let $f : \B^n \to \B^n$ be a Boolean AN such that $f(x)_i = x_i \lor
	x_j$ and $f(x)_j = x_i$, and $x \in \B^n$ such that $x_i = 0$ and $x_j = 1$. 
	We will compare $\mmjoblock{f}{\mu}(x)_i$ and $\mmjoblock{f}{\mu'}(x)_i$, in 
	order to prove a contradiction.
	Let us apply $\mmjoblock{f}{\mu}$ to $x$. 
	Before step $t_{\mu, i}$ the value of automaton $i$ is still $0$ and, most 
	importantly, since $t_{\mu, i} \leq t_{\mu,j}$, the value of $j$ is still $1$. 
	This means that right after step $t_{\mu,i}$, the value of automaton $i$ is 1, 
	and will not change afterwards. 
	Thus, we have $\mmjoblock{f}{\mu}(x)_i = 1$.
	Let us now apply $\mmjoblock{f}{\mu'}$ to $x$. 
	This time, $t_{\mu', i} > t_{\mu', j}$, which means that automaton $j$ is 
	updated first and takes the value of automaton $i$ at the time, which is $0$ 
	since it has not been updated yet. 
	Afterwards, neither automata will change value since $ 0 \lor 0$ is still $0$.
	This means that $\mmjoblock{f}{\mu'}(x)_i = 0$.
	Thus, we have $\mmjoblock{f}{\mu} \neq \mmjoblock{f}{\mu'}$, which contradicts
	our earlier hypothesis.
\end{proof}

Let $\BPn^0=\BPn/\equiv_0$ denote the corresponding quotient set, 
\emph{i.e.}~the set of block-parallel update modes to generate for computer 
analysis of all the possible dynamics in the general case of 
ANs
%BANs 
of size $n$.

\begin{theorem}
  \label{theorem:BPn0}
  For any $n\geq 1$, we have:
	\begin{align}
		|\BPn^0| 
			& \;=\; 
				\sum_{i = 1}^{p(n)} 
					\frac{n!}{\prod_{j = 1}^{d(i)} \left( m(i,j)! \right)^j}
					\label{BPn0_eq1}\\
			& \;=\; 
				\sum_{i = 1}^{p(n)} 
					\prod_{j = 1}^{d(i)} 
						\prod_{\ell = 1}^{j}
							\binom{n - \sum_{k = 1}^{j - 1} 
								k \cdot m(i,k) - (\ell - 1) \cdot m(i,j)}{m(i,j)}
								\label{BPn0_eq2}\\
			& \;=\; 
				\sum_{i = 1}^{p(n)} 
					\prod_{j = 1}^{d(i)} 
						\left( 
							\binom{n - \sum_{k = 1}^{j - 1} k \cdot m(i,k)}{j \cdot m(i,j)}
							\cdot 
							\prod_{\ell = 1}^{j} 
								\binom{(j - \ell + 1) \cdot m(i,j)}{m(i,j)}
						\right)\label{BPn0_eq3}\text{.}
	\end{align}
\end{theorem}

\begin{proof}
	$|\BPn^0|$ can be viewed as three distinct formulas. 
	We will show here that Formula~\ref{BPn0_eq1} counts $|\BPn^0|$,
	and the proof that these three formulas are equal will be in \ref{a:proofs}.

  For any pair $\mu, \mu' \in \BPn$, we have $\mu \equiv_0 \mu'$ if and only if
  their matrix-representations are the same up to a permutation of the elements 
  within columns (the number of equivalence classes is then counted by Formula~\ref{BPn0_eq1}).
  In the definition of $\varphi$, each block is a set constructed by taking one 
  element from each o-block. 
  Given that $n_k$ in the definition of $\varphi$ corresponds to $j$ in the 
  statement of the theorem, one matrix corresponds to all the o-blocks that have 
  the same size $n_k$.
  Hence, the $\ell \mod n_k$ operations in the definition of $\varphi$
  amounts to considering the elements of these o-blocks which are in the 
  same column in their matrix representation.
  Since blocks are unordered, the result follows.
\end{proof}

The first ten terms of the sequence $(|\BPn^0|)_{n\geq 1}$ are:
\[
	1, 3, 13, 67, 471, 3591, 33573, 329043, 3919387, 47827093\text{.}
\]
They match the sequence A182666 of OEIS~\cite{oeisA182666},
and the next lemma proves that they are indeed the same sequence
(defined by its exponential generating function on OEIS).
The \emph{exponential generating function of a sequence}
$\left(a_n\right)_{n\in\N}$ is $f(x) = \sum_{n\geq 0}a_n\frac{x^n}{n!}$.

\begin{lemma}
  \label{l:egf}
  The exponential generating function of $(|\BPn^0|)_{n \in \N}$ is
  $\prod_{j \geq 1} \sum_{k \geq 0} \left( \frac{x^k}{k!} \right)^j$.
\end{lemma}

The proof of this lemma can be found in \ref{a:proofs}.

As with the previous enumerating formula, we can also extrapolate an enumeration
algorithm from this one. It starts out similar to the previous one, but differs once the
contents of $M_j$ are chosen. In the previous algorithm, we needed to enumerate
every matrix up to permutation of the rows. This time, we need to enumerate
every matrix up to permutation within the columns. This means that, for each column
of the matrix, we just choose the content separately, similarly to how we chose the content of each matrix.

%%%%%%%%%%%%%%%%
\subsection{Partitioned orders up to dynamical isomorphism on the limit set}
\label{s:BPstar}

The following equivalence relation defined over block-parallel update modes
turns out to capture exactly the notion of having isomorphic limit dynamics.
It is analogous to $\equiv_0$, except that a circular shift of order $i$
may be applied on the sequences of blocks.

\begin{definition}
\label{def:dyn_iso}
  For $\mu, \mu' \in \BPn$, we denote $\mu \equiv_\star \mu'$ when 
  $\varphi(\mu) = \sigma^i(\varphi(\mu'))$ for some $i \in 
  \entiers{|\varphi(\mu')|}$ called the \emph{shift}.
\end{definition}

\begin{remark}\normalfont
  Note that $\mu \equiv_0 \mu'$ corresponds to the particular case $i = 0$ of 
  $\equiv_\star$.
  Thus, $\mu \equiv_0 \mu' \implies \mu \equiv_\star \mu'$.
\end{remark}

\begin{notation}
  Given $\mmjoblock{f}{\mu} : X \to X$, let
	$\Omega_{\mmjoblock{f}{\mu}} = \bigcap_{t \in \N} \mmjoblock{f}{\mu}(X)$
  denote its \emph{limit set} (abusing the notation of $\mmjoblock{f}{\mu}$ to 
  sets of configurations), 
  and $\mmjoblocklimit{f}{\mu} : \Omega_{\mmjoblock{f}{\mu}} \to
  \Omega_{\mmjoblock{f}{\mu}}$ its restriction to its limit set.
  Observe that, since the dynamics is deterministic, $\mmjoblocklimit{f}{\mu}$ 
  is bijective.
\end{notation}

The following Lemma shows that, if one is generally interested in the limit 
behavior of ANs under block-parallel updates, then studying a representative 
from each equivalence class of the relation $\equiv_\star$ is necessary and 
sufficient to get the full spectrum of possible limit dynamics.

\begin{lemma}
  For any $\mu, \mu' \in \BPn$, we have $\mu \equiv_\star \mu' \iff
  \forall f : X \to X, \mmjoblocklimit{f}{\mu} \sim \mmjoblocklimit{f}{\mu'}$.
\end{lemma}

\begin{proof}
	Let $\mu$ and $\mu'$ be two block-parallel update modes of $\BPn$.\medskip
	
  \noindent $(\Longrightarrow)$ Let $\mu,\mu'$ be such that $\mu \equiv_\star 
  \mu'$ of shift $\ispe \in \entiers{p}$, with $\varphi(\mu) = \Wl$, 
  $\varphi(\mu') = (W'_\ell)_{\ell \in \entiers{p}}$ and $p = |\varphi(\mu)| = 
  |\varphi(\mu')|$.
  It means that $\forall i \in \entiers{p}$, we have $W'_i = 
  W_{i + \ispe \mod p}$, and for any AN $f$, we deduce that 
  $\pi = \mmjblock{f}{W_0, \ldots, W_{\ispe-1}}$ is the desired isomorphism from 
  $\Omega_{\mmjoblock{f}{\mu}}$ to $\Omega_{\mmjoblock{f}{\mu'}}$.
  Indeed, we have $\mmjoblock{f}{\mu}(x) = y$ if and only if 
  $\mmjoblock{f}{\mu'}(\pi(x)) = \pi(y)$ because
  \[
    \mmjoblock{f}{\mu'} \circ \pi
    =
    \mmjblock{f}{W_0, \ldots, W_{\ispe-1}, W'_0, \ldots, W'_p}
    =
    \mmjblock{f}{W'_{p-\ispe}, \ldots, W'_p} \circ \mmjoblock{f}{\mu}
    =
    \pi \circ \mmjoblock{f}{\mu}\text{.}
  \]
  Note that $\pi^{-1} = \mmjoblock{f}{\mu'}^{(q-1)} \circ 
  \mmjblock{f}{W'_{\ispe} \ldots W'_{p-1}}$ with $q$ the least common multiple 
  of the limit cycle lengths, and $\pi^{-1} \circ \pi$ (\emph{resp.} 
  $\pi \circ\pi^{-1}$)  is the identity on $\Omega_{\mmjoblock{f}{\mu}}$
  (\emph{resp.} $\Omega_{\mmjoblock{f}{\mu'}}$).\medskip

  \noindent $(\Longleftarrow)$ We prove the contrapositive, from $\mu \not 
  \equiv_\star \mu'$, by case disjunction.\smallskip
  
  \begin{enumerate}
  \item[(1)] If in $\varphi(\mu)$ and $\varphi(\mu')$, there is an automaton 
	  $\ispe$ which is not updated the same number of times $\alpha$ and $\alpha'$ 
  	in $\mu$ and $\mu'$ respectively, then we assume without loss of generality 
	  that $\alpha > \alpha'$ and consider the AN $f$ such that:
  	\begin{itemize}[nosep]
    \item $X_{\ispe} = \entiers{\alpha}$ and $X_i = \{0\}$ for all 
    	$i \neq \ispe$; and
    \item $f_{\ispe}(x) = (x_{\ispe} + 1) \mod \alpha$ and $f_i(x) = x_i$ for 
    	all $i \neq \ispe$.
	  \end{itemize}
  	It follows that $\mmjoblocklimit{f}{\mu}$ has only fixed points since 
  	$+1 \mod \alpha$ is applied $\alpha$ times, whereas 
  	$\mmjoblocklimit{f}{\mu'}$ has no fixed point because $\alpha' < \alpha$.
	  We conclude that $\mmjoblocklimit{f}{\mu} \not \sim 
	  \mmjoblocklimit{f}{\mu'}$.\smallskip
	  
	\item[(2)] If in $\varphi(\mu)$ and $\varphi(\mu')$, all the automata are 
		updated the same number of times, then the transformation from $\mu$ to 
		$\mu'$ is a permutation on $\entiers{n}$ which preserves the matrices of 
		their matrix representations (meaning that any $i \in \entiers{n}$ is in an 
		o-block of the same size in $\mu$ and $\mu'$, which also implies that $\mu$ 
		and $\mu'$ are constructed from the same partition of $n$). 
		Then we consider subcases.
		
		\begin{enumerate}
		\item[(2.1)] If one matrix of $\mu'$ is not obtained by a permutation of the 
			columns from $\mu$, then there is a pair of automata $\ispe, \jspe$
			that appears in the $k$-th block of $\varphi(\mu)$ for some $k$,
			and does not appear in any block of $\varphi(\mu')$.
			Indeed, one can take $\ispe, \jspe$ to be in the same column in $\mu$ but 
			in different columns in $\mu'$.
			Let $S$ be the o-block of $\ispe$ and $S'$ be the o-block of $\jspe$.
			Let $p$ denote the least common multiple of o-blocks sizes in both $\mu$ 
			and $\mu'$.
			In this case we consider the AN $f$ such that:
			\begin{itemize}[nosep]
			\item $X_{\ispe} = \B \times \entiers{\frac{p}{|S|}}$,
	      $X_{\jspe} = \B \times \entiers{\frac{p}{|S'|}}$,
				and $X_i = \{0\}$ for all $i \notin \{\ispe,\jspe\}$.
				Given $x \in X$, we denote $x_{\ispe} = (x_{\ispe}^b, x_{\ispe}^\ell)$ 
				the state of $\ispe$ (and analogously for $\jspe$); and
			\item $f_{\ispe}(x) = \begin{cases}
        	(x_{\jspe}^b, x_{\ispe}^\ell + 1 \mod \frac{p}{|S|}) 
        		& \text{if } x_{\ispe}^\ell = 0\\
	        (x_{\ispe}^b, x_{\ispe}^\ell + 1 \mod \frac{p}{|S|}) 
  	      	& \text{otherwise}
    	  \end{cases}$,\\
	      $f_{\jspe}(x) = \begin{cases}
	        (x_{\ispe}^b, x_{\jspe}^\ell + 1 \mod \frac{p}{|S'|}) 
	        	& \text{if } x_{\jspe}^\ell = 0\\
        	(x_{\jspe}^b, x_{\jspe}^\ell + 1 \mod \frac{p}{|S'|}) 
        		& \text{otherwise}
	      \end{cases}$, and\\
	      $f_i(x) = x_i$ for all $i \notin \{\ispe,\jspe\}$.
 			\end{itemize}
			Note that $\ispe$ (resp.~$\jspe$) is updated $\frac{p}{|S|}$ 
			(resp.~$\frac{p}{|S'|}$) times during a step in both $\mu$ and $\mu'$.
			Therefore for any $x \in X$, its two images under $\mu$ and $\mu'$ verify
			$\mmjoblock{f}{\mu}(x)^\ell_{\ispe} = \mmjoblock{f}{\mu'}(x)^\ell_{\ispe} 
				= x^\ell_{\ispe}$ (and analogously for $\jspe$).
			Thus for the evolution of the states of $\ispe$ and $\jspe$ during a step,
			the second element is fixed and only the first element (in $\B$) may 
			change.
			We split $X$ into $X^{=} = \{x \in X \mid x_{\ispe}^b = x_{\jspe}^b\}$
			and $X^{\neq} = \{x \in X \mid x_{\ispe}^b \neq x_{\jspe}^b\}$,
			and observe the following facts by the definition of $f_{\ispe}$ and 
			$f_{\jspe}$:
			\begin{itemize}[nosep]
			\item Under $\mu$ and $\mu'$, all the elements of $X^=$ are fixed points
				(indeed, only $x_{\ispe}^b$ and $x_{\jspe}^b$ may evolve by copying the 
				other).
			\item Under $\mu$, let $m, m'$ be the respective number of times $\ispe,
				\jspe$ have been updated prior to the $k$-th block of $\varphi(\mu)$ in 
				which they are updated synchronously.
				Consider the configurations $x, y \in X^{\neq}$ with $x_{\ispe} = 
				(0, -m \mod \frac{p}{|S|})$, $x_{\jspe} = (1, -m' \mod \frac{p}{|S'|})$,
				$y_{\ispe} = (1, -m \mod \frac{p}{|S|})$ and $y_{\jspe} = (0, -m' \mod	
				\frac{p}{|S'|})$.
				It holds that $\mmjoblock{f}{\mu}(x) = y$ and 
				$\mmjoblock{f}{\mu}(y) = x$, because $x_{\ispe}^b$ and $x_{\jspe}^b$ are 
				exchanged synchronously when $x_{\ispe}^\ell = x_{\jspe}^\ell = 0$ 
				during the $k$-th block of $\varphi(\mu)$, and are not exchanged again 
				during that step by the choice of the modulo.
				Hence, $\mmjoblocklimit{f}{\mu}$ has a limit cycle of length two.
			\item Under $\mu'$, for any $x \in X^{\neq}$, there is a substep with 
				$x_{\ispe}^\ell = 0$ and there is a substep with $x_{\jspe}^\ell = 0$,
				but they are not the same substep (because $\ispe$ and $\jspe$ are never 
				synchronized in $\mu'$).
				As a consequence, $x_{\ispe}^b$ and $x_{\jspe}^b$ will end up having the 
				same value (the first to be updated copies the bit from the second, then 
				the second copies its own bit), \emph{i.e.}~$\mmjoblock{f}{\mu'}(x) \in 
				X^{=}$, and therefore $\mmjoblocklimit{f}{\mu'}$ has only fixed points.
		  \end{itemize}
			We conclude in this case that $\mmjoblocklimit{f}{\mu} \not \sim
			\mmjoblocklimit{f}{\mu'}$, because one has a limit cycle of length two, 
			whereas the other has only fixed points.

		\item[(2.2)] If the permutation preserves the columns within the matrices
	  (meaning that the automata within the same column in $\mu$ are also in the 
	  same column in $\mu'$), then we consider two last subcases:
	  	\begin{enumerate}
	  	\item[(2.2.1)] Moreover, if the permutation of some matrix is not 
	  		circular (meaning that there are three columns which are not in the same 
	  		relative order in $\mu$ and $\mu'$), then there are three automata 
	  		$\ispe$, $\jspe$ and $\kspe$ in the same matrix such that in $\mu$, 
	  		automaton $\ispe$ is updated first, then $\jspe$, then $\kspe$; whereas
				in $\mu'$, automaton $\ispe$ is updated first, then $\kspe$, then 
				$\jspe$. 
				Let us consider the automata network $f$ such that:
				\begin{itemize}[nosep]
				\item $X = \B^n$;
				\item $f_{\ispe}(x) = x_{\kspe}$, $f_{\jspe}(x) = x_{\ispe}$ and 
					$f_{\kspe}(x) = x_{\jspe}$; and
				\item $f_i(x) = x_i$ if $i \notin \{\ispe, \jspe, \kspe\}$.
				\end{itemize}
				If the three automata are updated in the order $\ispe$ then $\jspe$ 
				then $\kspe$, as it is the case with $\mu$, then after any update, 
				they will all have taken the same value.
				It implies that $\mmjoblock{f}{\mu}$ has only fixed points,
				precisely the set 
				$P = \{x \in \B^n \mid x_{\ispe} = x_{\jspe} = x_{\kspe}\}$.

				If they are updated in the order $\ispe$ then $\kspe$ then $\jspe$, as
				with $\mu'$, however, the situation is a bit more complex.
				We consider two cases, according to the number of times they are updated 
				during a period (recall that since they belong to the same matrix,
				they are updated repeatedly in the same order during the substeps):
				\begin{itemize}[nosep]
				\item If they are updated an odd number of times each, then automata $
					\ispe$ and $\jspe$ will take the initial value of automaton $\kspe$, 
					and automaton $\kspe$ will take the initial value of automaton 
					$\jspe$.
					In this case, $\mmjoblocklimit{f}{\mu'}$ has the fixed points $P$
					and limit cycles of length two.
				\item If they are updated an even number of times each, then the reverse 
					will occur: 
					automata $\ispe$ and $\jspe$ will take the initial value of automaton 
					$\jspe$, and automaton $\kspe$ will keep its initial value.
					In this case, $\mmjoblocklimit{f}{\mu'}$ has the fixed points 
					$Q = \{x \in \B^n \mid x_{\ispe} = x_{\jspe}\}$ which strictly 
					contains $P$ (\emph{i.e.}~$P \subseteq Q$ and $Q \setminus P \neq
					\emptyset$).
				\end{itemize}
				In both cases $\mmjoblocklimit{f}{\mu'}$ has more than the fixed points 
				$P$ in its limit set, hence we conclude that
				$\mmjoblocklimit{f}{\mu} \not \sim \mmjoblocklimit{f}{\mu'}$.
	  	\item[(2.2.2)] Moreover, if the permutation of all matrices is circular,
				then we first observe that when $\varphi(\mu)$ and $\varphi(\mu')$ have 
				one block in common, they have all blocks in common (because of the 
				circular nature of permutations), \emph{i.e.}~$\mu \equiv_\star \mu'$.
				Thus, under our hypothesis, we deduce that $\varphi(\mu)$ and 
				$\varphi(\mu')$ have no block in common.
				As a consequence, there exist automata $\ispe, \jspe$ with the property 
				from case~(2.1), namely synchronized in a block of 
				$\varphi(\mu)$ but never synchronized in any block of $\varphi(\mu')$,
			  and the same construction terminates this proof.
			\end{enumerate}
		\end{enumerate}
  \end{enumerate}
\end{proof}

Let $\BPn^\star = \BPn / \equiv_\star$ denote the corresponding quotient set.

\begin{theorem}
	Let $\lcm(i) = \lcm(\{j \in \llbracket 1, d(i) \rrbracket \mid 
	m(i,j) \geq 1\})$.
	For any $n \geq 1$, we have:
	\[
		|\BPn^\star| = 
			\sum_{i = 1}^{p(n)} 
				\frac{n!}{\prod_{j = 1}^{d(i)} 
					\left( 
						m(i,j)! 
					\right)^j}
				\cdot 
				\frac{1}{\lcm(i)}\text{.}
  \]
\end{theorem}

\begin{proof}
	Let $\mu, \mu' \in \BPn$ two update modes such that $\mu \equiv \mu'$. 
	Then their sequential forms are of the same length, and each automaton 
	appears the same number of times in both of them. 
	This means that, if an automaton is in an o-block of size $k$ in $\mu$'s 
	partitioned order form, then it is also in an o-block of the same size in 
	$\mu'$'s. 
	We deduce that two update modes of size $n$ can only be equivalent as defined 
	in Definition~\ref{def:dyn_iso} if they are generated from the same partition 
	of $n$.

	Let $\mu \in \BPn^0$, generated from partition $i$ of $n$.
	Then $\varphi(\mu)$ is of length
%$\lcm(\{j \in \llbracket 1, d(i)\rrbracket | m(i,j)
%\geq 1\})$, that we will denote
	$\lcm(i)$. 
	Since no two elements of $\BPn^0$ have the same block-sequential form, the 
	equivalence class of $\mu$ in $\BPn^0$ contains exactly $\lcm(i)$ elements, 
	all generated from the same partition $i$ (all the blocks of $\varphi(\mu)$ 
	are different).
	Thus, the number of elements of $\BPn^\star$ generated from a partition $i$
	is the number of elements of $\BPn^0$ generated from partition $i$,
	divided by the number of elements in its equivalence class for $\BPn^\star$, 
	namely $\lcm(i)$.
\end{proof}

\begin{remark}\normalfont\label{remark:BPnstar}
  The formula for $|\BPn^\star|$ can actually be obtained from any formula in
  Theorem~\ref{theorem:BPn0} by multiplying by $\frac{1}{\lcm(i)}$ inside the 
  sum on partitions (from $i = 1$ to $p(n)$).
\end{remark}

While counting the elements of $\BPn^\star$ was pretty straightforward,
enumerating them by ensuring that no two partitioned orders are the same up 
to circular permutation of their block-sequential rewritings 
(Definition~\ref{def:dyn_iso}) is more challenging.
This is performed by Algorithm~\ref{algo:BPiso}.
It works much like the one enumareting the elements of $\BPn^0$, except for the following differences. 
In the first function, \texttt{EnumBPiso}:
right after choosing the partition, a list of coefficients $a[j]$ is determined, 
in a modified algorithm that computes $\lcm(i)$ inductively (lines 3-10). 
These coefficients are used in the second auxiliary function 
\texttt{EnumBlockIsoAux}, where the minimum $min_j$ of the matrix $M_j$ is 
forced to be in the first $a[j]$ columns of said matrix (lines 35-37, the 
condition is fulfilled when $min_j$ has not been chosen within the $a[j] - 1$ 
first columns, then it is placed in that column so only $m(i,j) - 1$ elements 
are chosen).

\begin{algorithm}[t!]
  {\scriptsize \caption{Enumeration of $\BPn^\star$}
  \label{algo:BPiso}
  \DontPrintSemicolon
  %\SetAlgoLined
  \SetKw{KwAnd}{and}
  \SetKwProg{Fn}{Function}{:}{}
  \SetKwFunction{FEnumBP}{EnumBPiso}
  \SetKwFunction{FEnumBPaux}{EnumBPisoAux}
  \SetKwFunction{FEnumBlock}{EnumBlockIso}
  \SetKwFunction{FEnumBlockAux}{EnumBlockIsoAux}
  \SetKwComment{Comment}{\# }{}
	\bigskip  
  
  \Fn{\FEnumBP{$n$}}{
    \ForEach{$i \in \emph{partitions}(n)$}{
      $a$ is a list of size $d(i)$\;
      $b \leftarrow 1$\;
      \For{$j \leftarrow d(i)$ \KwTo $1$}{
        \If{$m(i,j) > 0$}{
          $a[j] \leftarrow \gcd(b, j)$\;
          $b \leftarrow \lcm(b, j)$\;
        }
        \Else{
          $a[j] \leftarrow j$\;
        }
      }
      \FEnumBPaux{$n$, $i$, $1$, $a$}
    }
  }\medskip
  
  \Fn{\FEnumBPaux{$n$, $i$, $j$, $a$, $M_1$, \dots, $M_{j-1}$}}{
    \If{$d(i) < j$}{
      \texttt{enumerate(}$M$\texttt{)}\;
      \Return
    }
    \If{$m(i, j) > 0$}{ % not mandatory...
      \ForEach{combination $A$ of size $j \cdot m(i, j)$ among
				$\entiers{n} \setminus \bigcup_{k = 1}^{j - 1} M_k$}{
        $min_j \leftarrow min(A)$\;
        \ForEach{$M_j$ enumerated by \FEnumBlock{$A$, $j$, $m(i,j)$,
					$min_j$, $a \emph{[} j \emph{]}$}}{
          \FEnumBPaux{$n$, $i$, $j+1$, $a$, $M_1$, \dots, $M_{j-1}$, $M_j$}
        }
      }
    }
    \Else{
      \FEnumBPaux{$n$, $i$, $j+1$, $a$, $M_1$, \dots, $M_{j-1}$, $\emptyset$}
    }
  }\medskip

  \Fn{\FEnumBlock{$A$, $j$, $m$, $min_j$, $a_j$}}{
    \ForEach{$C$ enumerated by \FEnumBlockAux{$A$, $m$, $min_j$, $a_j$}}{
      \For{$k \leftarrow 1$ \KwTo $m$}{
        \For{$\ell \leftarrow 1$ \KwTo $j$}{
          $M_j[k][\ell] = C_\ell[k]$\;
        }
      }
      \texttt{enumerate(}$M_j$\texttt{)}\;
      \Return
    }
  }\medskip

  \Fn{\FEnumBlockAux{$A$, $j$, $m$, $min_j$, $a_j$, $C_1$, ..., $C_\ell$}}{
    \If{$A = \emptyset$}{
      \texttt{enumerate(}$C$\texttt{)}\;
      \Return
    }
    \Else{
      \If{$|A| = m \cdot(j - a_j + 1)$ \KwAnd $min_j \in A$}{
        \ForEach{combination $B$ of size $m - 1$ among $A \setminus \{min_j\}$}{
          \FEnumBlockAux{$A \setminus (B \cup \{min_j\})$, $m$, $min_j$, $a_j$,
						$C_1$, \dots, $C_\ell$, $(B \cup \{min_j\})$}\;
        }
      }
      \Else{
        \ForEach{combination $B$ of size $m$ among $A$}{
          \FEnumBlockAux{$A \setminus B$, $m$, $min_j$, $a_j$, $C_1$, \dots,
						$C_\ell$, $B$}\;
        }
      }
    }
  }}
\end{algorithm}

The proof of correction of Algorithm~\ref{algo:BPiso} can be found in \ref{a:proofs}.

%%%%%%%%%%%%%%%%
\subsection{Implementations}
\label{s:implementations}

Proof-of-concept Python implementations of the three
enumeration algorithms mentioned
are available on the following repository:
\begin{center}
  \url{https://framagit.org/leah.tapin/blockpargen}.
\end{center}
It is archived by Software Heritage at the following permalink:
\begin{center}
  \href{https://archive.softwareheritage.org/browse/directory/f1b4d83c854a4d042db5018de86b7f41ef312a07/?origin_url=https://framagit.org/leah.tapin/blockpargen}{https://archive.softwareheritage.org/browse/directory/\\f1b4d83c854a4d042db5018de86b7f41ef312a07/?origin\_url=\\https://framagit.org/leah.tapin/blockpargen}.
\end{center}
We have conducted numerical experiments on a standard laptop,
presented on Figure~\ref{fig:xp}.

\begin{figure}
\begin{minipage}{0.4\textwidth}
 \centering
 \begin{tabular}{| c | c | c | c |}
  \hline
  $n$ & $\BPn$ & $\BPn^0$ & $\BPn^\star$ \\
  \hline
  1 & 1 & 1 & 1\\
  & - & - & -\\
  \hline
  2 & 3 & 3 & 2\\
  & - & - & -\\
  \hline
  3 & 13 & 13 & 6\\
  & - & - & -\\
  \hline
  4 & 73 & 67 & 24\\
  & - & - & -\\
  \hline
  5 & 501 & 471 & 120\\
  & - & - & -\\
  \hline
  6 & 4051 & 3591 & 795\\
  & - & - & -\\
  \hline
  7 & 37633 & 33573 & 5565\\
  & - & 0.103s & -\\
  \hline
  8 & 394353 & 329043 & 46060\\
  & 0.523s & 0.996s & 0.161s\\
  \hline
  9 & 4596553 & 3919387 & 454860\\
  & 6.17s & 12.2s & 1.51s\\
  \hline
  10 & 58941091 & 47827093 & 4727835\\
  & 1min24s & 2min40s & 16.3s\\
  \hline
  11 & 824073141 & 663429603 & 54223785\\
  & 21min12s & 38min31s & 3min13s\\
  \hline
  12 & 12470162233 & 9764977399 & 734932121\\
  & 5h27min38s & 9h49min26s & 45min09s\\
  \hline
  \end{tabular}
\end{minipage}\hfill
\begin{minipage}{0.42\textwidth}
\centering
\includegraphics[scale=0.43]{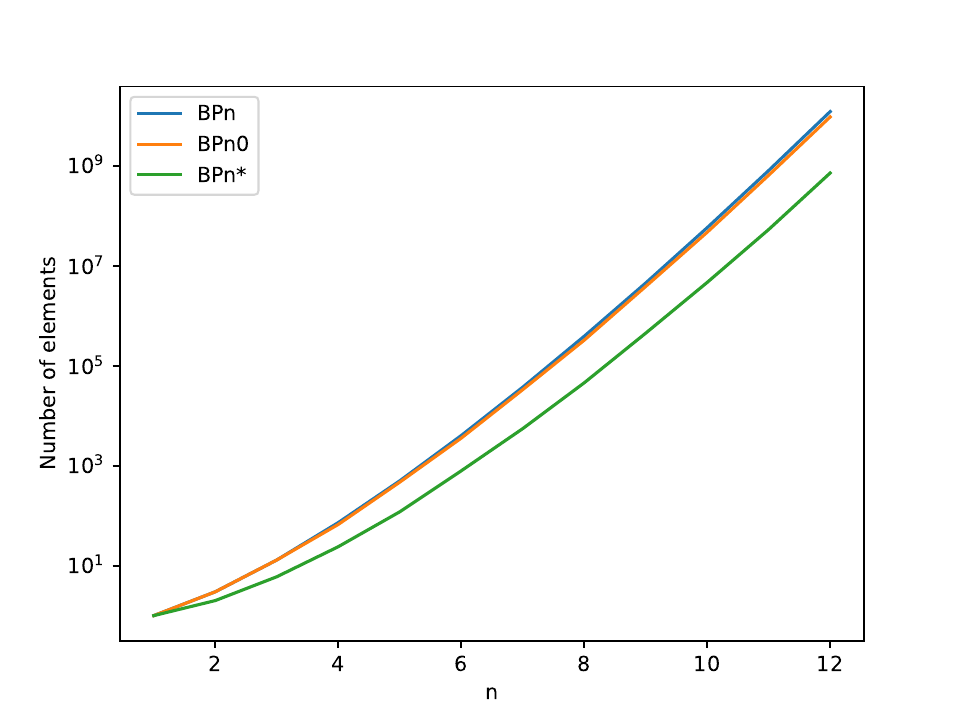}
\end{minipage}
\caption{
  Numerical experiments of our Python implementation of
  the three algorithms
  on a standard laptop (processor Intel-Core$^\text{TM}$ i7 @ 2.80 GHz).
  For $n$ from $1$ to $12$, the table (left) presents the size of
  $\BPn$, $\BPn^0$ and $\BPn^*$ and running time to enumerate
  their elements (one representative of each equivalence class;
  a dash represents a time smaller than $0.1$ second),
  and the graphics (right) depicts their respective sizes
  on a logarithmic scale.
  Observe that the sizes of $\BPn$ and $\BPn^0$ are comparable,
  whereas an order of magnitude is gained with $\BPn^*$,
  which may be significant for advanced numerical experiments
  regarding limit dynamics under block-parallel udpate modes.
}
\label{fig:xp}
\end{figure}

%%%%%%%%%%%%%%%%% COMPLEXITY

\section{Computational complexity under block-parallel updates}
\label{s:complexity}

Computational complexity is important to anyone willing to use algorithmic tools
in order to study discrete dynamical systems.
Lower bounds inform on the best worst case time or space one can expect
with an algorithm solving some problem.
The $n$ local functions of a BAN are encoded as Boolean circuits,
which is a convenient formalism corresponding to the high level descriptions one usually employs.
The update mode is given as a list of lists of integers, each of them being encoded either in unary or binary
(this makes no difference, because the encoding of local functions already has a size greater than $n$).

In this section we characterize the computational complexity of typical problems
arising in the framework of automata networks.
We will see that almost all problems reach $\PSPACE$-completeness.
The intuition behind this fact is that the description of a block-parallel update mode
may expend (through $\varphi$) to an exponential number of substeps,
during which a linear bounded Turing machine may be simulated via iterations of a circuit.
We first recall this folklore building block and present a general outline of our constructions
(Subsection~\ref{ss:construction}).
Then we start with results on computing images, preimages, fixed points and limit cycles
(Subsection~\ref{ss:image}),
before studying reachability and global properties of the function $\mmjoblock{f}{\mu}$
computed by an automata network $f$ under block-parallel update schedule $\mu$
(Subsection~\ref{ss:reach}).

%%%%%%%%%%%%%%%%
\subsection{Outline of the $\PSPACE$-hardness constructions}
\label{ss:construction}

We will design polynomial time many-one reductions from the following
$\PSPACE$-complete decision problem, which appears for example in~\cite{J-Goles2016}.\\[.5em]
\decisionpb{Iterated Circuit Value Problem}{Iter-CVP}
{a Boolean circuit $C:\B^n\to\B^n$, a configuration $x\in\B^n$, and $i\in\entiers{n}$.}
{does $\exists t\in\N:C^t(x)_i=\1$?}

\begin{theorem}[folklore]
  {\normalfont \textbf{Iter-CVP}} is $\PSPACE$-complete.
\end{theorem}

Before presenting the general outline of our constructions,
we need a technical lemma related to the generation of primes (proof in \ref{a:proofs}).

\begin{lemma}\label{lem:primes}
  For all $n \geq 2$, a list of distinct prime integers
  $p_1,p_2,\dots,p_{k_n}$ such that $2\leq p_i < n^2$ and $2^n < \prod_{i=1}^{k_n} p_i < 2^{2n^2}$
  can be computed in time $\O(n^2)$, with $k_n=\lfloor\frac{n^2}{2\ln(n)}\rfloor$.
\end{lemma}

Our constructions of automata netwoks and block-parallel update schedules
for the computational complexity lower bounds are based on the following.

\begin{definition}\label{def:gn}
  For any $n\geq 2$, let $p_1,p_2,\dots,p_{k_n}$ be the $k_n$ primes given by Lemma~\ref{lem:primes},
  and denote $q_j=\sum_{i=1}^{j}p_i$ their cumulative series for $j$ from $0$ to $k_n$.
  Define the automata network $g_n$ on $q_{k_n}$ automata $\entiers{q_{k_n}}$ with constant $\0$ local functions,
  where the components are grouped in o-blocks of length $p_i$, that is with
  $\mu_n=\bigcup_{i\in\entiers{k_n}}\{(q_i,q_i+1,\dots,q_{i+1}-1)\}$.
\end{definition}

\begin{lemma}\label{lem:gn}
  For any $n\geq 2$, one can compute $g_n$ and $\mu_n$ in time $\O(n^4)$,
  and $|\varphi(\mu_n)|>2^n$.
\end{lemma}

\begin{proof}
  The time bound comes from Lemma~\ref{lem:primes} and the fact that $q_{k_n}$ is in $\O(n^4)$.
  The number of blocks in $\varphi(\mu_n)$ is the least common multiple of its o-block sizes,
  which is the product $\prod_{i=1}^{k_n} p_i$, hence from Lemma~\ref{lem:primes}
  we conclude that it is greater than $2^n$.
\end{proof}

The general idea is now to add some automata to $g_n$
and place them within singletons in $\mu_n$, i.e., each of them in a new o-block of length $1$.
We propose an example implementing a binary counter on $n$ bits.

\begin{example}\label{ex:counter}
  Given $n\geq 2$, consider $g_n$ and $\mu_n$ given by Lemma~\ref{lem:gn}.
  Construct $f$ from $g_n$ by adding $n$ Boolean components $\{q_{k_n},\dots,q_{k_n+n}\}$,
  whose local functions increment a binary counter on those $n$ bits,
  until it freezes to $2^n-1$ (all bits in state $\1$).
  Construct $\mu'$ from $\mu_n$ as $\mu'=\mu_n\cup\bigcup_{i\in\entiers{n}}\{(q_{k_n}+i)\}$,
  so that the counter components are updated at each substep.
  Observe that the pair $f,\mu'$ can be still be computed from $n$ in time $\O(n^4)$.
  Figure~\ref{fig:counter} illustrates an example of orbit for $n=3$,
  and one can notice that $\mmjoblock{f}{\mu'}$ is a constant function sending any $x\in\B^n$ to $\0^{q_{k_n}}\1^n$.
  \begin{figure}
    \centering
    \includegraphics{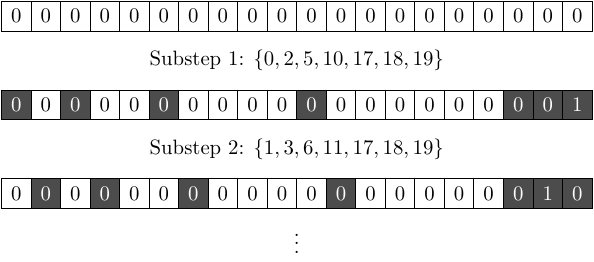}
    \caption{
      Substeps leading to the image of configuration $\0^{q_{k_n}}\0\1\0$ in $\mmjoblock{f}{\mu'}$
      from Example~\ref{ex:counter} for $n=3$ ($k_n=4$ and $q_{k_n}=2+3+5+7=17$).
      The last $3$ bits implement a binary counter, freezing at $7$ ($\1\1\1$).
      Above each substep the block of updated automata is given.
    }
    \label{fig:counter}
  \end{figure}
\end{example}

Remark that we will prove complexity lower bounds by reduction from \textbf{Iter-CVP},
where $n$ will be the number of inputs and outputs of the circuit to be iterated,
hence the integer $n$ itself will be encoded in unary.
As a consequence, the construction of Example~\ref{ex:counter} is computed in polynomial time.

%%%%%%%%%%%%%%%%
\subsection{Images, preimages, fixed points and limit cycles}
\label{ss:image}

We start the study of the computational complexity of automata networks under block-parallel update schedules
with the most basic problem of computing the image $\mmjoblock{f}{\mu}(x)$
of some configuration $x$ through $\mmjoblock{f}{\mu}$
(i.e., one step of the evolution),
which is already $\PSPACE$-hard.
We conduct this study as decision problems.
It is actually hard to compute even a single bit of $\mmjoblock{f}{\mu}(x)$.
The fixed point verification problem is a particular case of computing an image, which is still $\PSPACE$-hard
(unlike block-sequential update schedules for which this problem is in $\Poly$).
Recall that the encoding of $\mu$ (with integers in unary or binary) has no decisive influence on the input size,
this latter being characterized by the circuits sizes and in particular their number of inputs,
denoted $n$, which is encoded in unary.
\\[.5em]
\decisionpb{Block-parallel step bit}{BP-Step-Bit}
{$(f_i:\B^n\to\B)_{i\in\entiers{n}}$ as circuits, $\mu\in\BPn$, $x\in\B^n$, $j\in\entiers{n}$.}
%{a BAN of size $n$ given as $n$ local functions $f_i:\B^n\to\B$ for $i\in\entiers{n}$,
%a block-parallel schedule $\mu\in\BPn$, a configuration $x\in\B^n$ and an integer $j\in\entiers{n}$.}
{does $\mmjoblock{f}{\mu}(x)_j=1$?}
\\[.5em]
\decisionpb{Block-parallel step}{BP-Step}
{$(f_i:\B^n\to\B)_{i\in\entiers{n}}$ as circuits, $\mu\in\BPn$, $x,y\in\B^n$.}
{does $\mmjoblock{f}{\mu}(x)=y$?}
\\[.5em]
\decisionpb{Block-parallel fixed point verification}{BP-Fixed-Point-Verif}
{$(f_i:\B^n\to\B)_{i\in\entiers{n}}$ as circuits, $\mu\in\BPn$, $x\in\B^n$.}
{does $\mmjoblock{f}{\mu}(x)=x$?}\\

This first set of problems is related to the image of a given configuration $x$,
which allows the reasonings to concentrate on the dynamics of substeps for that single configuration $x$,
regardless of what happens for other configurations.
Note that $n$ will be the size of the \textbf{Iter-CVP} instance,
while the size of the automata network will be $q_{k_n}+\ell'+n+1$.

\begin{theorem}\label{thm:image}
  {\normalfont \textbf{BP-Step-Bit}}, {\normalfont \textbf{BP-Step}} and {\normalfont \textbf{BP-Fixed-Point-Verif}}
  are $\PSPACE$-complete.
\end{theorem}

\begin{proof}
  The problems \textbf{BP-Step-Bit}, \textbf{BP-Step} and \textbf{BP-Fixed-Point-Verif}
  are in $\PSPACE$, with a simple algorithm obtaining $\mmjoblock{f}{\mu}(x)$ by
  computing the least common multiple of o-block sizes and
  then using a pointer for each block throughout
  the computation of that number of substeps
  (each substep evaluates local functions in polynomial time).

  \medskip

  We give a single reduction for the hardness of \textbf{BP-Step-Bit}, \textbf{BP-Step}
  and \textbf{BP-Fixed-Point-Verif},
  where we only need to consider the dynamics of the substeps starting from one configuration $x$.
  Given an instance of \textbf{Iter-CVP} with a circuit $C:\B^n\to\B^n$,
  a configuration $\tilde{x}\in\B^n$ and $i\in\entiers{n}$,
  we apply Lemma~\ref{lem:gn} to construct $g_n,\mu_n$ on automata set $P=\entiers{q_{k_n}}$.
  Automata from $P$ have constant $\0$ local functions,
  and the number of substeps is $\ell=|\varphi(\mu_n)|>2^n$
  thanks to the prime's $\lcm$.
  We define a BAN $f$ by adding:
  \begin{itemize}[nosep]
    \item $\ell'=\lceil\log_2(\ell)\rceil$ automata numbered $B=\{q_{k_n},\dots,q_{k_n}+\ell'-1\}$,
      implementing a counter that increments modulo $\ell$ at each substep,
      and remains fixed when $x_B$ encodes an integer greater or equal to $\ell$
      (case not considered in this proof);
    \item $n$ automata numbered $D=\{q_{k_n}+\ell',\dots,q_{k_n}+\ell'+n-1\}$, whose local functions
      iterate $C:\B^n\to \B^n$ while the counter is smaller than $\ell-1$,
      and go to state $\tilde{x}$ when the counter reaches $\ell-1$, i.e.,~with
      \[
        f_D(x)=
        \begin{cases}
          C(x_D)&\text{if }x_B<\ell-1\text{,}\\
          \tilde{x}&\text{otherwise; and}
        \end{cases}
      \]
    \item $1$ automaton numbered $R=\{q_{k_n}+\ell'+n\}$, whose local function
      \[
        f_R(x)=x_R\vee x_{q_{k_n}+\ell'+i}
      \]
      records whether a state $\1$
      appeared at automaton in relative position $i$ within $D$.
  \end{itemize}
  We also add singletons to $\mu_n$ for each of these additional automata, by setting
  \[
    \mu'=\mu_n\cup\bigcup_{j\in B\cup D\cup R}\{(j)\}.
  \]
  Now, consider the dynamics of substeps in computing the image of configuration
  %\[
    $x=\0^{q_{k_n}}\0^{\ell'}\tilde{x}\0$.
  %\]
  During the first $\ell-1$ substeps:
  \begin{itemize}[nosep]
    \item automata $P$ have constant $\0$ local function;
    \item automata $B$ increment a counter from $0$ to $\ell-1$;
    \item automata $D$ iterate circuit $C$ from $\tilde{x}$; and
    \item automaton $R$ records whether the $i$-th bit of $D$ has been in state $\1$ during some iteration.
  \end{itemize}
  During the last substep,
  automata $B$ go back to $\0^n$ because of the modulo, and
  automata $D$ go back to state $\tilde{x}$.
  Since the number of substeps $\ell$ is greater than $2^n$ (Lemma~\ref{lem:gn}),
  the iterations of $C$ search the whole orbit of $\tilde{x}$,
  and at the end of the step automaton $R$ has recorded whether the \textbf{Iter-CVP} instance
  is positive (went to state $\1$) or negative (still in state $\0$).
  The images are respectively
  $y_-=\0^{q_{k_n}}\0^{\ell'}\tilde{x}\0$ or $y_+=\0^{q_{k_n}}\0^{\ell'}\tilde{x}\1$.
  This concludes the reductions, to \textbf{BP-Step-Bit} by asking whether automaton $R$
  (numbered $q_{k_n}+2n$) is in state $\1$,
  to \textbf{BP-Step} by asking whether the image of $x$ is $y_+$,
  and to \textbf{BP-Fixed-Point-Verif} because $y_-=x$ ($\coPSPACE$-hardness).
\end{proof}

As a corollary, the associated functional problem of computing $\mmjoblock{f}{\mu}$
is computable in polynomial space and is $\PSPACE$-hard for polynomial time Turing reductions
(not for many-one reductions, as there is no concept of negative instance for total functional problems).
%\begin{corollary}\label{cor:image}
%  Given an automata network $(f_i:\B^n\to\B)_{i\in\entiers{n}}$ as circuits,
%  a block-parallel update schedule $\mu\in\BPn$ and a configuration $x\in\B^n$,
%  the problem of computing $\mmjoblock{f}{\mu}(x)$ is in $\FPSPACEPoly$ and $\PSPACE$-hard.
%\end{corollary}
Deciding whether a given configuration $y$ has a preimage through $\mmjoblock{f}{\mu}$ is also $\PSPACE$-complete
(see \ref{a:proofs} for details).
%Since $\mmjoblock{f}{\mu}$ is not enforced to be injective
%(this will be considered in Section~\ref{ss:reach}), $\mmjoblock{f}{\mu}^{-1}(y)$ is a set.

Now, we study the computational complexity of problems related to the existence
of fixed points and limit cycles in an automata network under block-parallel update schedule.
Again, we need to consider the image of all configurations,
and have no control on neither the start configuration $x$ nor the end configuration $y$
during the dynamics of substeps.
In particular, the counter may be initialized to any value,
and the bit $R$ may already be set to $\1$.
We adapt the previous reductions accordingly.\\[.5em]
\decisionpb{Block-parallel fixed point}{BP-Fixed-Point}
{$(f_i:\B^n\to\B)_{i\in\entiers{n}}$ as circuits, $\mu\in\BPn$.}
{does $\exists x\in\B^n:\mmjoblock{f}{\mu}(x)=x$?}\\[.5em]
\decisionpb{Block-parallel limit cycle of length $k$}{BP-Limit-Cycle-$k$}
{$(f_i:\B^n\to\B)_{i\in\entiers{n}}$ as circuits, $\mu\in\BPn$.}
{does $\exists x\in\B^n:\mmjoblock{f}{\mu}^k(x)=x$?}\\[.5em]
\decisionpb{Block-parallel limit cycle}{BP-Limit-Cycle}
{$(f_i:\B^n\to\B)_{i\in\entiers{n}}$ as circuits, $\mu\in\BPn$, $k\in\N_+$.}
{does $\exists x\in\B^n:\mmjoblock{f}{\mu}^k(x)=x$?}\\[.5em]
On limit cycles we have a family of problems (one for each integer $k$),
and a version where $k$ is part of the input (encoded in binary).
It makes no difference on the complexity.

\begin{theorem}
  \label{thm:fixedpoint}
  {\normalfont \textbf{BP-Fixed-Point}},
  {\normalfont \textbf{BP-Limit-Cycle-$k$}} for any $k\in\N_+$
  and {\normalfont \textbf{BP-Limit-Cycle}} are $\PSPACE$-complete.
\end{theorem}

\begin{proof}
  These problems still belong to $\PSPACE$, because they amount to enumerating configurations
  and computing images by $\mmjoblock{f}{\mu}$,
  which can be performed from \textbf{BP-Step}
  (Theorem~\ref{thm:image}).

  \medskip

  We start with the hardness proof for the fixed point existence problem,
  and we will then adapt it to limit cycle existence problems.
  Given an instance $C:\B^n\to\B^n$, $\tilde{x}\in\B^n$, $i\in\entiers{n}$ of \textbf{Iter-CVP},
  we construct the same block-parallel update schedule $\mu'$ as in the proof of Theorem~\ref{thm:image},
  and modify the local functions of automata $B$ and $R$ as follows:
  \begin{itemize}[nosep]
    \item automata $B$ increment a counter modulo $\ell$ at each substep,
      and go to $0$ when the counter is greater than (or equal to) $\ell-1$; 
      and
    \item automaton $R$ records whether a state $\1$ appears at the $i$-th bit of $x_D$,
      and flips when the counter is equal to $\ell-1$, i.e.,\\
        \hspace*{5.5cm}$f_R(x)=
        \begin{cases}
          x_R\vee x_{q_{k_n}+\ell'+i}&\text{if }x_B<\ell-1\text{,}\\
          \neg x_R&\text{otherwise.}
        \end{cases}$\\[-.2em]
      %\[
      %  f_R(x)=
      %  \begin{cases}
      %    x_R\vee x_{q_{k_n}+\ell'+i}&\text{if }x_B<\ell-1\\
      %    \neg x_R&\text{otherwise}
      %  \end{cases}
      %\]
  \end{itemize}
  Recall that automata $D$ iterate the circuit when $x_B<\ell-1$ and go to $\tilde{x}$ otherwise,
  and that the number $\ell$ of substeps is larger than $2^n$.

  If the \textbf{Iter-CVP} instance is positive, then configuration
  $x=\0^{q_{k_n}}\0^{\ell'}\tilde{x}\0$ is a fixed point of $\mmjoblock{f}{\mu'}$.
  Indeed, during the $\ell$-th and last substep,
  the primes $P$ are still in state $\0^{q_{k_n}}$, 
  the counter $B$ goes back to $0$ (state $\0^{\ell'}$),
  the circuit $D$ goes back to $\tilde{x}$,
  and automaton $R$ has recorded the $\1$ which is flipped into state $\0$.

  Conversely, if there is a fixed point configuration $x$,
  then the counter must be at most $\ell-1$ because of the modulo $\ell$ increment.
  Furthermore, automata $D$ will encounter one substep during which it goes to $\tilde{x}$,
  hence the resulting configuration on $D$ will be in the orbit of $\tilde{x}$,
  i.e., $x_D$ is in the orbit of $\tilde{x}$.
  Finally, automaton $R$ will also encounter exactly one substep during which it
  is flipped (when $x_B\geq \ell-1$).
  As a consequence, in order to go back to its initial value $x_R$,
  the state of $R$ must be flipped during another substep,
  which can only happen when it is in state $\0$ and automaton $q_{k_n}+\ell'+i$
  is in state $\1$.
  We conclude that the $i$-th bit of a configuration in the orbit of $\tilde{x}$
  is in state $\1$ during some iteration of the circuit $C$,
  meaning that the \textbf{Iter-CVP} instance is positive.
  Remark that in this case, configuration $\0^{q_{k_n}}\0^{\ell}\tilde{x}\0$ is one of the fixed points.

  \medskip

  For the limit cycle existence problems,
  we modify the construction to let the counter go up to $k\ell-1$.
  Precisely:
  \begin{itemize}[nosep]
    \item $\ell'=\lceil\log_2(k\ell)\rceil$ automata $B$ implement a binary counter
      which is incremented at each substep, and goes to $0$ when $x_B\geq k\ell-1$;
    %\item $n$ automata $D$ iterate the circuit $C$ when $x_B<\ell-1$,
    %  and go to state $\tilde{x}$ otherwise (no change),
    \item $n$ automata $D$ iterate the circuit $C$ if $x_B<\ell-1$,
      else go to state $\tilde{x}$ (no change); and 
    \item $1$ automaton $R$ records whether a state $\1$ appears in the $i$-th bit of $x_D$,
      and flips when the counter is equal to $\ell-1$.
  \end{itemize}
  The reasoning is identical to the case $k=1$, except that the counter needs $k$ times $\ell$ substeps,
  i.e., $k$ steps, in order to go back to its initial value.
  As a consequence, there is no $x$ and $k'<k$ such that $\mmjoblock{f}{\mu}^{k'}(x)=x$,
  and the dynamics has no limit cycle of length smaller than $k$.
  Remark that when the \textbf{Iter-CVP} instance is positive,
  configurations $(\0^{q_{k_n}}B_i\tilde{x}\0)_{i\in\entiers{k}}$
  with $B_i$ the $\ell'$-bits encoding of $i\ell$
  form one of the limit cycles of length $k$.
  Also remark that the encoding of $k$ in binary within the input has no consequence,
  neither on the $\PSPACE$ algorithm, nor on the polynomial time many-one reduction.
\end{proof}

Remark that our construction also applies to the notion of limit cycle $x^0,\dots,x^{p-1}$
where it is furthermore required that all configurations are different
(this corresponds to having the minimum length $p$): the problem is still $\PSPACE$-complete.

%%%%%%%%%%%%%%%%
\subsection{Reachability and general complexity bounds}
\label{ss:reach}

In this part, we settle the computational complexity of the classical reachability problem,
which is unsurprisingly still $\PSPACE$-hard by reduction from another model of computation
(see \ref{a:proofs} for details).
In light of what precedes, one may be inclined to think that any problem
related to the dynamics of automata networks under block-parallel update schedules is $\PSPACE$-hard.
We prove that this is partly true with a general complexity bound theorem
on subdynamics existing within $\mmjoblock{f}{\mu}$, based on our previous results on fixed points and limit cycles.
However, we will also prove that a Rice-like complexity lower bound analogous
to the main results of~\cite{C-Gamard2021}, i.e., which would state that
any non-trivial question on the dynamics
(on the functional graph of $\mmjoblock{f}{\mu}$)
expressible in first order logics
is $\PSPACE$-hard, does not hold
(unless a collapse of $\PSPACE$ to the first level of the polynomial hierarchy).
Indeed, we will see that deciding the bijectivity ($\forall x,y\in\B^n:\mmjoblock{f}{\mu}(x)=\mmjoblock{f}{\mu}(y)\implies x=y$) is complete for $\coNP$.
We conclude the section with a discussion on reversible dynamics.%\\[.5em]

From the fixed point and limit cycle theorems in Section~\ref{ss:image},
we now derive that any particular subdynamics is hard to identify within $\mmjoblock{f}{\mu}$
under block-parallel update schedule.
A functional graph is a directed graph of out-degree exactly one,
and we assimilate $\mmjoblock{f}{\mu}$ to its functional graph.
We define a family of problems, one for each functional graph $G$
to find as a subgraph of $\mmjoblock{f}{\mu}$,
and prove that the problem is always $\PSPACE$-hard.
Since $\PSPACE=\coPSPACE$, checking the existence of a subdynamics
is as hard as checking the absence of a subdynamics,
even though the former is a local property whereas the latter is a global property
at the dynamics scale. This is understandable in regard of the fact that
$\PSPACE$ scales everything to the global level
(one can search the whole dynamics in $\PSPACE$),
because verifying that a given set of configurations (a certificate)
gives the subgraph $G$ is difficult (Theorem~\ref{thm:image}).\\[.5em]
\decisionpb{Block-parallel $G$ as subdynamics}{BP-Subdynamics-$G$}
{$(f_i:\B^n\to\B)_{i\in\entiers{n}}$ as circuits, $\mu\in\BPn$.}
{does $G\sqsubset\mmjoblock{f}{\mu}$?}\\[.5em]
Remark that asking whether $G$ appears as a subgraph or as an induced subgraph
makes no difference when $G$ is functional (has out-degree exactly one),
because $\mmjoblock{f}{\mu}$ is also functional:
it is necessarily induced since there is no arc to delete.

\begin{theorem}\label{thm:subdynamics}
  {\normalfont \textbf{BP-Subdynamics-$G$}} is $\PSPACE$-complete for any functional graph $G$.
\end{theorem}

\begin{proof}
  A polynomial space algorithm for \textbf{BP-$G$-Subdynamics} consists in
  enumerating all subsets $S\subseteq\B^n$ of size $|S|=|V(G)|$,
  and test for each whether the restriction of $\mmjoblock{f}{\mu}$ to $S$ is isomorphic to $G$
  (functional graphs are planar hence isomorphism can be decided in logarithmic space~\cite{C-Datta2010}).

  \medskip

  For the $\PSPACE$-hardness, the idea is to choose a fixed point or limit cycle in $G$,
  and make it the decisive element whose existence or not lets $G$ be a subgraph
  of the dynamics or not.
  Since $G$ is a functional graph, it is composed of fixed points and limit cycles,
  with hanging trees rooted into them (the trees are pointing towards their root).
  Let $G(v)$ denote the unique out-neighbor of $v\in V(G)$.

  Let us first assume that $G$ has a limit cycle of length $k\geq 2$,
  or a fixed point with a tree of height greater or equal to $1$ hanging
  (the case where $G$ has only isolated limit cycles is treated thereafter).
  A fixed point is assimilated to a limit cycle of length $k=1$.
  Let $G'$ be the graph $G$ without this limit cycle of size $k$,
  and let $U$ be the vertices of $G'$ without out-neighbor (if $k=1$ then $U\neq\emptyset$).
  We reduce from \textbf{Iter-CVP}, and first compute the $f,\mu$ of size $n$ obtained
  by the reduction from Theorem~\ref{thm:fixedpoint}
  for the problem \textbf{BP-Limit-Cycle-$k$}.
  We have that
  $\mmjoblock{f}{\mu}$ has a limit cycle of length $k$
  on configurations $(\0^{q_{k_n}}B_i\tilde{x}\0)_{i\in\entiers{k}}$
  (or configuration $\0^{q_{k_n}}\0^\ell\tilde{x}\0$ for $k=1$)
  if and only if the \textbf{Iter-CVP} instance is positive.

  We construct $g$ on $n+1$ automata,
  and the update schedule $\mu'$ being the union of $\mu$ with a singleton o-block for the new automaton.
  We assume that $n\geq |V(G)|-k$,
  otherwise we pad $f,\mu$ to that size (with identity local functions for the new automata).
  The idea is that $g$ will consist in a copy of $f$ on the subspace $x_n=\0$,
  and a copy of $G'$ on the subspace $x_n=\1$ where the images of the configurations
  corresponding to the vertices of $U$ will be configurations of the potential limit cycle
  of $\mmjoblock{f}{\mu}$ (in the other subspace $x_n=\0$).
  Other configurations in the subspace $x_n=\1$ will be fixed points.
  Figure~\ref{fig:subdynamics} illustrates the construction.
  Recall that $G$ is fixed,
  and consider a mapping $\alpha:V(G)\to\bool^n$
  such that vertices of the limit cycle of length $k$ are sent
  to the configurations $(\0^{q_{k_n}}B_i\tilde{x}\0)_{i\in\entiers{k}}$ respectively
  (or $\0^{q_{k_n}}\0^\ell\tilde{x}\0$ for $k=1$).
  We define:\\
  %\[
    \hspace*{3cm}
    $g(x)=\begin{cases}
      f(x_{\entiers{n}})\0 & \text{if } x_n=0\text{,}\\
      \alpha(G(v))\0 & \text{if } x_n=1 \text{ and } \exists v\in U:\alpha(v)=x_{\entiers{n}}\text{,}\\
      \alpha(G(v))\1 & \text{if } x_n=1 \text{ and } \exists v\in G'\setminus U:\alpha(v)=x_{\entiers{n}}\text{,}\\
      x & \text{otherwise.}
    \end{cases}$\\[.4em]
  %\]
  %\begin{itemize}[nosep]
  %  \item The $n$ first automata have local functions
  %    \[
  %      g_{\entiers{n}}(x)=\TODO{wrong: correct it!}
  %      \begin{cases}
  %        f(x)&\text{if } x_n=\0\\
  %        \0&\text{if $x_n=\1$ and the configuration is in $U$}\\
  %        \1&\text{otherwise}
  %      \end{cases}
  %    \]
  %    that is, they copy the dynamics of $f$ when the new automaton is in state $\0$,
  %    and otherwise they remain fixed (except in the particular case of $U$).
  %  \item The local function $g_n$ of the new automaton remains fixed when it is in state $\0$,
  %    and uses $|V(G')|=|V(G)|-k$ configurations with $x_n=\1$
  %    (for example with $x_{\entiers{n}}$ from $0$ to $|V(G')|-1$ in binary)
  %    in order to construct a copy of $G'$:
  %    these configurations are either wired among themselves if they do not belong to $U$,
  %    or to the corresponding vertex in the potential cycle $(\0^{q_{k_n}}B_i\tilde{x}\0)_{i\in\entiers{k}}$
  %    of $\mmjoblock{f}{\mu}$
  %    (or the potential fixed point $\0^{q_{k_n}}\0^\ell\tilde{x}\0$ in the case $k=1$).
  %    %Remark that since we use these details in the proof of Theorem~\ref{thm:fixedpoint},
  %    %to be precise we are presenting a reduction from \textbf{Iter-CVP}.
  %    The remaining configurations with $x_n=\1$ remain fixed.
  %\end{itemize}
  \begin{figure}
    \centering
    \includegraphics{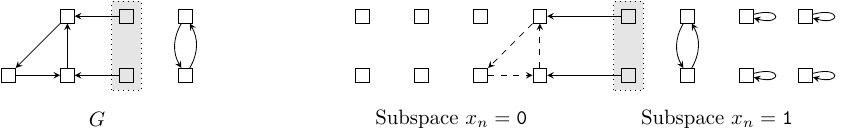}
    \caption{
      Construction of $g$ in the proof of Theorem~\ref{thm:subdynamics}.
      Subspace $x_n=\0$ contains a copy of $f$
      with a potential limit cycle dashed.
      Subspace $x_n=\1$ implements $G'$, and wires configurations of $U$ (grey area)
      to the potential limit cycle in the copy of $f$
      (remaining configurations are fixed points).
    }
    \label{fig:subdynamics}
  \end{figure}
  The obtained dynamics $\mmjoblock{g}{\mu'}$ has one copy of $\mmjoblock{f}{\mu}$
  (in subspace $x_n=\0$), with a copy of $G'$ (in subspace $x_n=\1$)
  which becomes a copy of $G$ if configurations
  $(\0^{q_{k_n}}B_i\tilde{x}\0)_{i\in\entiers{k}}$ (or $\0^{q_{k_n}}\0^\ell\tilde{x}\0$ in the case $k=1$)
  form a limit cycle of length $k$.
  Moreover, it becomes a copy of $G$ only if so by our assumption on the limit cycle or fixed point of $G$,
  because the remaining configurations in subspace $x_n=\1$ are all isolated fixed points.
  This concludes the reduction.

  For the case where $G$ is made of $k$ isolated fixed points,
  we reduce from \textbf{BP-Fixed-Point} and construct an automata network with
  $k$ copies of the dynamics of $f$,
  by adding $\lceil\log_2(k)\rceil$ automata with identity local functions.
\end{proof}

When the property of being a functional graph is dropped,
that is when the out-degree of $G$ is at most one
(otherwise any instance is trivially negative),
problem \textbf{BP-Subdynamics-$G$} is subtler.
Indeed, one can still ask for the existence of fixed points, limit cycles
and any functional subdynamics $\PSPACE$-complete by Theorem~\ref{thm:subdynamics},
but new problems arise, some of which are provably complete only for $\coNP$.
The symmetry of existence versus non existence is broken.
In what follows, we settle that deciding the bijectivity of $\mmjoblock{f}{\mu}$ is $\coNP$-complete,
and then discuss the complexity of decision problems which are subsets of bijective networks,
such as the problem of deciding whether $\mmjoblock{f}{\mu}$ is the identity.
We conclude the section by proving that it is nevertheless $\PSPACE$-complete to decide whether
$\mmjoblock{f}{\mu}$ is a constant map.
These results hint at the subtleties behind a full characterization of the computational complexity
of \textbf{BP-Subdynamics-$G$} for all graphs of out-degree at most one.\\[.5em]
\decisionpb{Block-parallel bijectivity}{BP-Bijectivity}
{$(f_i:\B^n\to\B)_{i\in\entiers{n}}$ as circuits, $\mu\in\BPn$.}
{is $\mmjoblock{f}{\mu}$ bijective?}\\[.5em]
Remark that, because the space of configurations is finite,
injectivity, surjectivity and bijectivity are equivalent properties of $\mmjoblock{f}{\mu}$.

\begin{lemma}\label{lem:bij}
  Let $f:\B^n\to\B^n$ a BAN and $\mu\in\BPn$ a block-parallel update mode.
  Then $\mmjoblock{f}{\mu}$ is bijective if and only if
  $\mmjblock{f}{W}$ is bijective for every block $W$ of $\varphi(\mu)$.
\end{lemma}

\begin{proof}
  The right to left implication is obvious since $\mmjoblock{f}{\mu}$
  is a composition of bijections $\mmjblock{f}{W}$.
  We prove the contrapositive  of the left to right implication,
  assuming the existence of a block $W$ in $\varphi(\mu)$
  such that $\mmjblock{f}{W}$ is not bijective.
  Let $W_\ell$ be the first such block in the sequence $\varphi(\mu)$,
  so there exist $x,y\in\B^n$ such that $x\neq y$ but $\mmjblock{f}{W_\ell}(x)=\mmjblock{f}{W_\ell}(y)=z$.
  By minimality of $\ell$, the composition $g=\mmjblock{f}{W_{\ell-1}}\circ\dots\circ\mmjblock{f}{W_0}$
  is bijective, hence there also exist $x',y'\in\B^n$ with $x'\neq y'$ such that $g(x')=x$ and $g(y')=y$.
  That is, after the $\ell$-th substep the two configurations $x'$ and $y'$ have the same image $z$,
  and we conclude that $\mmjoblock{f}{\mu}(x')=\mmjoblock{f}{\mu}(y')=
  \mmjblock{f}{W_{p-1}}\circ\dots\circ\mmjblock{f}{W_{\ell+1}}(z)$
  therefore $\mmjoblock{f}{\mu}$ is not bijective.
\end{proof}

Lemma~\ref{lem:bij} shows that bijectivity can be decided at the local level of circuits
(not iterated), which can be checked in $\coNP$ and gives Theorem~\ref{thm:bij}.

\begin{theorem}
  \label{thm:bij}
  {\normalfont \textbf{BP-Bijectivity}} is $\coNP$-complete.
\end{theorem}

\begin{proof}
  A $\coNP$ algorithm can be established from Lemma~\ref{lem:bij},
  because it is equivalent to check the bijectivity at all substeps.
  A non-deterministic algorithm can guess a temporality
  $t\in\entiers{|\varphi(\mu)|}$ (in binary) within the substeps,
  two configurations $x,y$,
  and then check in polynomial time that they certify the non-bijectivity of that substep as follows.
  First, construct $W$ the $t$-th block of $\varphi(\mu)$,
  by computing $t$ modulo each o-block size to get the automata from that o-block.
  Second, check that $\mmjblock{f}{W}(x)=\mmjblock{f}{W}(y)$.

  The $\coNP$-hardness is a direct consequence of that complexity lower bound
  for the particular case of the parallel update schedule~\cite[Theorem~5.17]{HDR-Perrot2022}.
\end{proof}

We now turn our attention to the recognition of identity dynamics.\\[.5em]
\decisionpb{Block-parallel identity}{BP-Identity}
{$(f_i:\B^n\to\B)_{i\in\entiers{n}}$ as circuits, $\mu\in\BPn$.}
{does $\mmjoblock{f}{\mu}(x)=x$ for all $x\in\B^n$?}\\[.5em]
This problem is in $\PSPACE$, and is $\coNP$-hard by reduction
from the same problem in the parallel case~\cite[Theorem~5.18]{HDR-Perrot2022}.
However, it is neither obvious to design a $\coNP$-algorithm to solve it,
nor to prove $\PSPACE$-hardness by reduction from \textbf{Iter-CVP}.

\begin{open}\label{open:id}
  {\normalfont \textbf{BP-Identity}} is $\coNP$-hard and in $\PSPACE$.
  For which complexity class is it complete?
\end{open}

A major obstacle to the design of an algorithm, or of a reduction
from \textbf{Iter-CVP} to \textbf{BP-Identity}, lies in the fact
that, by Theorem~\ref{thm:bij}, ``hard'' instances of the latter are bijective networks
(because non-bijective instances can be recognized in our immediate lower bound $\coNP$,
and they are all negative instances of \textbf{BP-Identity}).
A reduction would therefore be related to the lengths of cycles in the dynamics of substeps,
and whether they divide the least common multiple of o-block sizes
(for $x\in\B^n$ such that $f(x)=x$) or not ($f(x)\neq x$).

Nonetheless, we are able to prove another lower bound, 
related to the hardness of computing the number of models of a given propositional formula.
%The canonical $\ModkPoly{k}$-complete problem consists in deciding whether this quantity is
%congruent to $1$ modulo $k$ (the case $k=2$ is denoted $\PPoly$,
%it is closed by complementary when $k$ is prime).
The canonical $\ModPoly$-complete problem takes as input a formula $\psi$
and two integers $k$, $i$ encoded in unary, and consists in deciding whether
the number of models of $\psi$ is congruent to $k$ modulo the $i$-th prime number
(which can be computed in polytime).
It generalizes classes $\ModkPoly{k}$ (such as the parity case $\ModkPoly{2}=\PPoly$),
and it is notable that $\SPoly$ polytime truth-table reduces to $\ModPoly$~\cite{J-Kobler1996}.

%by reduction from \textbf{Parity-SAT}
%which is the canonical $\PPoly$-complete problem.
%By Valiant-Vazirani theorem, any problem in $\NP$ is randomized many-one reducible to \textbf{Parity-SAT}
%(denoted $\NP\subseteq\textsf{RP}^{\PPoly}$),
%but $\NP$ and $\PPoly$ are not known to be comparable~\cite{J-Dell2013}.
%% KP: source: https://cs.stackexchange.com/questions/128603/is-even-sat-np-hard

\begin{theorem}\label{thm:idmodp}
  {\normalfont \textbf{BP-Identity}} is $\ModPoly$-hard (for polytime many-one reduction).
\end{theorem}

The proof of this theorem can be found in \ref{a:proofs}.

Our attemps to prove $\PSPACE$-hardness failed, for the following reasons.
To get bijective circuits one could reduce from reversible Turing machines (RTM) and problem
\textbf{Reversible Linear Space Acceptance}~\cite{J-Lange2000}.
A natural strategy would be to simulate a RTM for an exponential number of subteps,
and then simulate it backwards for that same number of substeps,
while ending in the exact same configuration (identity map) if and only if
the simulation did not halt or was not in the orbit of the given input $w$.
The difficulty with this approach is that the dynamics of substeps must not be the identity map
when a \textit{conjunction} of two temporally separated events happens:
first that the simulation has halted, and second that the starting configuration was $w$.
It therefore requires to remember at least one bit of information,
which is subtle in the reversible setting.
Indeed, the constructions of~\cite{J-Lange2000} and~\cite{B-Morita2017} %KP: chapters 2 and 8
consider only starting configurations of the Turing machine in the initial state and with blank tapes.
However, in the context of Boolean automata networks,
any configuration must be considered (hence any configuration of the simulated Turing machine).

Regarding iterated circuits simulating reversible cellular automata
(for which the whole configuration space is usualy considered),
the literature focuses on
decidability issues~\cite{C-Kari2005,J-Sutner2004},
%KP: Introduction (page 1-2):
% ``The combination of these two ideas, that reversible logic can simulate combinational logic, and
% that iterated combinational logic characterizes PSPACE, would suggest that iterated reversible logic
% might again be as powerful as PSPACE. However, this is not obvious and does not follow from the
% known simulations of combinational logic by reversible logic. The problem is that these simulations
% use additional bits of input and output, beyond those of the simulated circuit. Padding bits of input
% (assumed to be zero) are transformed by the reversible circuits of these simulations into garbage bits
% of output (not necessarily zero).''
but a recent contribution fits our setting
and we derive the following.
$\FPoPSPACE$ is the class of functions computable in polynomial time
with an oracle in $\PSPACE$.

\begin{theorem}[{\cite[Theorem~5.7]{A-Eppstein2023}}]\label{thm:rca}
  There is a one-dimensional reversible cellular automaton for which simulating
  any given number of iterations, with periodic boundary conditions, is complete for $\FPoPSPACE$
\end{theorem}

\begin{corollary}
    \label{cor:fpspacecomplete}
  Given $(f_i:\B^n\to\B)_{i\in\entiers{n}}$ as circuits, $\mu\in\BPn$ such that $\mmjoblock{f}{\mu}$ is bijective,
  $x\in\B^n$ and $t\in\entiers{|\varphi(\mu)|}$ in binary,
  computing the configuration at the $t$-th substep is complete for $\FPoPSPACE$.
\end{corollary}

The proof of this corollary can be found in \ref{a:proofs}.

Intuitively, the dynamics of substeps embeds complexity.
The relationship to the complexity of computing the configuration after the whole step
composed of $|\varphi(\mu)|$ substeps, \emph{i.e.}~the image through $\mmjoblock{f}{\mu}$,
%in order to reach \textbf{BP-Identity},
is not obvious.

\medskip

Being a constant map is another global property of the dynamics,
which turns out to be $\PSPACE$-complete to recognize for BANs
under block-parallel update schedules.\\[.5em]
\decisionpb{Block-parallel constant}{BP-Constant}
{$(f_i:\B^n\to\B)_{i\in\entiers{n}}$ as circuits, $\mu\in\BPn$.}
{does there exist $y\in\B^n$ such that $\mmjoblock{f}{\mu}(x)=y$ for all $x\in\B^n$?}

\begin{theorem}
  \label{thm:cst}
  {\normalfont \textbf{BP-Constant}} is $\PSPACE$-complete.
\end{theorem}

\begin{proof}
  To decide \textbf{BP-Constant}, one can simply enumerate all configurations
  and compute their image (Theorem~\ref{thm:image}) while checking that it always gives the same result.

  \medskip

  For the $\PSPACE$-hardness proof, we reduce from \textbf{Iter-CVP}.
  Given a circuit $C:\B^n\to\B^n$, a configuration $\tilde{x}$ and $i\in\entiers{n}$,
  we apply Lemma~\ref{lem:gn} to construct $g_n,\mu_n$ on automata set $P=\entiers{q_{k_n}}$.
  Automata from $P$ have constant $\0$ local functions,
  and the number of substeps is $\ell=|\varphi(\mu_n)|>2^n$.
  We add (Figure~\ref{fig:cst} illustrates the obtained dynamics):
  \begin{itemize}[nosep]
    \item $\ell'=\lceil\log_2(\ell)\rceil$ automata numbered $B=\{q_{k_n},\dots,q_{k_n}+\ell'-1\}$,
      implementing a $\ell'$-bits binary counter that increments at each substep,
      and sets all automata from $B$ in state $\1$ when the counter is greater or equal to $\ell-1$;
    \item $n$ automata numbered $D=\{q_{k_n}+\ell',\dots,q_{k_n}+\ell'+n-1\}$, whose local functions
      are given below; and 
    \item $1$ automaton numbered $R=\{q_{k_n}+\ell'+n\}$, whose local function is given below.
  \end{itemize}
      \[
        f_D(x)=
        \begin{cases}
          C(\tilde{x})&\text{if }x_B=0\\
          C(x_D)&\text{if }0<x_B<\ell-1\\
          \0^n&\text{otherwise}
        \end{cases}
        \qquad
        f_R(x)=
        \begin{cases}
          \tilde{x}_i&\text{if }x_B=0\\
          x_R\vee x_{q_{k_n}+\ell'+i}&\text{if }0<x_B<\ell\\
          \1&\text{otherwise}
        \end{cases}
      \]
  We also add singletons to $\mu_n$ for these additional automata, via
  %\[
    $\mu'=\mu_n\cup\bigcup_{j\in B\cup D\cup R}\{(j)\}$.
  %\]

  \begin{figure}
    \centering
    \includegraphics[width=\textwidth]{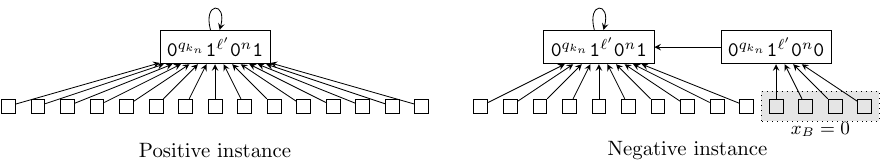}
    \caption{
      Illustration of the dynamics obtained for the reduction
      to \textbf{BP-Constant} in the proof of Theorem~\ref{thm:cst}.
      Configurations $x$ with the counter automata $B$ initialized to $x_B=0$
      either go to $\0^{q_{k_n}}\1^{\ell'}\0^n\1$
      (left, positive instance),
      or to $\0^{q_{k_n}}\1^{\ell'}\0^n\0$ (right, negative instance).
      Only the bit of automata $R$ changes.
    }
    \label{fig:cst}
  \end{figure}

  For any configuration $x$ with a counter not initialized to $0$,
  i.e.,~with $x_B\neq 0$, the counter will reach and remain in the all $\1$ state
  before the last substep, therefore automata from $D$ will be updated to $\0^n$
  and automaton $R$ will be updated to $\1$.
  We conclude that $\mmjoblock{f}{\mu'}(x)=\0^{q_{k_n}}\1^{\ell'}\0^n\1$.
  For configurations $x$ with $x_B=0$, substeps proceed as follows:
  \begin{itemize}[nosep]
    \item automata $B$ count until $\ell-1$ at the penultimate substep
      (recall that $\ell=|\varphi(\mu_n)|=|\varphi(\mu'_n)|$),
      which finally brings them all in state $\1$ during the last substep;
    \item automata $D$ iterate the circuit $C$,
      starting from $C(\tilde{x})$ during the first substep; and 
    \item automaton $R$ records whether a $\1$ appears or not in the whole orbit of $\tilde{x}$
      (recall that $\ell=|\varphi(\mu'_n)|>2^n$),
      starting from $\tilde{x}$ itself during the first substep
      (even though $x_D\neq\tilde{x}$)
      and without encountering the ``$\1$ otherwise'' case.
  \end{itemize}
  We conclude that the image of $x$ on automata $P$ is $\0^{q_{k_n}}$,
  on $B$ is $\1^{\ell'}$, on $D$ is $\0^n$,
  and on $R$ it depends whether the \textbf{Iter-CVP} instance is positive
  (automaton $R$ in state $\1$) or negative (automaton $R$ in state $\0$).
  This completes the reduction: the image is always $\0^{q_{k_n}}\1^{\ell'}\0^n\1$
  if and only if the \textbf{Iter-CVP} instance is positive.
\end{proof}

%%%%%%%%%%%%%%%%% Conclusion

\section{Conclusion and perspectives}
\label{s:conclusion}

This article presents a theoretical study of block-parallel update modes
under two different angles.

The first part of this article focused on combinatorial aspects, in particular 
counting and enumerating not only the general set of block-parallel update mode,
but also the classes for two equivalence relations. These relations weren't chosen
arbitrarily, since the first one gives us the minimal number of update modes required
to generate every possible distinct underlying dynamical system, and the second one
generates dynamical systems that all have different sets of limit cycles (limit dynamics).
Regarding the enumeration algorithms we presented, one of the first questions that
comes to mind would be about their complexity. The three algorithms we mentioned seem to belong
to \textsf{EnumP}~\cite{J-Creignou2019}, but this requires a formal analysis.
In particular, the question
of which subclasses of \textsf{EnumP} they belong to still needs to be addressed.

The second half of this article didn't answer this question, rather focusing on the
computational complexity of classical decision problems involving Boolean automata networks.
While an automaton cannot be updated twice in a block-sequential update mode, hence limiting the number
of substeps, update repetitions are allowed in block-parallel update modes, which leads to a much
higher ceiling for the number of substeps. This provides a greater expressiveness, but also
higher complexity costs.
Mainly, computing a single transition goes from being feasible in polynomial time with all
block-sequential schedules~\cite{C-Perrotin2023} to $\PSPACE$-hard in this context
(Theorem~\ref{thm:image}). This raises the complexity of classical problems related to the
existence of preimages, fixed points, limit cycles, and the recognition of constant dynamics
from $\NP$-complete (existence problems) or $\coNP$-complete (global dynamical properties)
for block-sequential schedules to $\PSPACE$-complete for block-parallel schedules.

One might be tempted to draw the following conjecture from these results.
\begin{conjecture}[false]\label{conj:pspace}
  If a problem is $\NP$-hard or $\coNP$-hard and in $\PSPACE$
  for block-sequential update schedules
  then it is $\PSPACE$-complete for block-parallel update schedules.
\end{conjecture}

This conjecture is however disproven by the recognition of bijective dynamics. Since a
single substep is enough to identify the absence of bijectivity, the increase in the 
number of substeps did not affect the complexity of the problem, and it stays
$\coNP$-complete for block-parallel schedules.
The complexity of recognizing identity dynamics under block-parallel schedules is still
an open problem.

After determining the complexity of recognizing preimages, image points or
fixed points, the next logical step would be the complexity of counting them.

An important remark is that, while these proofs were written for Boolean
automata networks, they also apply to multi-valued automata networks.

As mentioned earlier, the possibility of updating the same automaton multiple times in one step
allows for a greater variety of possible dynamics. Especially, they can break the fixed point
invariance property that holds for every block-sequential update mode,
but may fail for block-parallel update mode with such repetitions. It would thus be pertinent to
characterize what in these repetitions triggers the creation of new fixed points.
More generally, it would be interesting to study the following problem: given an AN $f$ ,
to which extent is $f$ block-parallel sensible or robust?
In~\cite{J-Aracena2009,J-Aracena2012,J-Aracena2013}, the authors addressed this
question on block-sequential Boolean ANs by developing the concept of update digraphs,
capturing conditions of dynamical equivalence at the syntactical level.
This concept is however unapplicable to update schedules where automata repetitions appear.
Creating an equivalent of update digraphs that allows for update repetitions would be an 
important step in the understanding of updating sensitivity and robustness of ANs.
Another approach would be to study specifically how interaction cycles behave when paired with
block-parallel update modes.

\paragraph{Acknowledgements}

This work received support from ANR-18-CE40-0002 FANs, STIC AmSud CAMA 
22-STIC-02 (Campus France MEAE) and HORIZON-MSCA-2022-SE-01-101131549 
ACANCOS projects.

\bibliographystyle{plainurl}

\begin{thebibliography}{10}

\bibitem{oeisA000262}
OEIS A000262.
\newblock The online encyclopedia of integer sequences.
\newblock \url{https://oeis.org/A000262}.

\bibitem{oeisA000670}
OEIS A000670.
\newblock The online encyclopedia of integer sequences.
\newblock \url{https://oeis.org/A000670}.

\bibitem{oeisA061095}
OEIS A061095.
\newblock The online encyclopedia of integer sequences.
\newblock \url{https://oeis.org/A061095}.

\bibitem{oeisA182666}
OEIS A182666.
\newblock The online encyclopedia of integer sequences.
\newblock \url{https://oeis.org/A182666}.

\bibitem{C-Akutsu1999}
T.~Akutsu, S.~Miyano, and S.~Kuhara.
\newblock {Identification of genetic networks from a small number of gene
  expression patterns under the Boolean network model}.
\newblock In {\em Proceedings of the Pacific Symposium on Biocomputing}, pages
  17--28. Wold Scientific, 1999.

\bibitem{J-Alon1985}
N.~Alon.
\newblock Asynchronous threshold networks.
\newblock {\em Graphs and Combinatorics}, 1:305--–310, 1985.
\newblock \href {https://doi.org/10.1007/BF02582959}
  {\path{doi:10.1007/BF02582959}}.

\bibitem{J-Aracena2013b}
J.~Aracena, J.~Demongeot, {\'E}.~Fanchon, and M.~Montalva.
\newblock {On the number of different dynamics in Boolean networks with
  deterministic update schedules}.
\newblock {\em Mathematical Biosciences}, 242:188--194, 2013.
\newblock \href {https://doi.org/10.1016/j.mbs.2013.01.007}
  {\path{doi:10.1016/j.mbs.2013.01.007}}.

\bibitem{J-Aracena2011}
J.~Aracena, \'E. Fanchon, M.~Montalva, and M.~Noual.
\newblock {Combinatorics on update digraphs in Boolean networks}.
\newblock {\em Discrete Applied Mathematics}, 159:401--409, 2011.

\bibitem{J-Aracena2012}
J.~Aracena, {\'E}.~Fanchon, M.~Montalva, and M.~Noual.
\newblock {ombinatorics on update digraphs in Boolean networks}.
\newblock {\em Discrete Applied Mathematics}, 159:401--409, 2012.

\bibitem{J-Aracena2009}
J.~Aracena, E.~Goles, A.~Moreira, and L.~Salinas.
\newblock {On the robustness of update schedules in Boolean networks}.
\newblock {\em Biosystems}, 97:1--8, 2009.

\bibitem{J-Aracena2013}
J.~Aracena, L.~G{\'o}mez, and L.~Salinas.
\newblock {Limit cycles and update digraphs in Boolean networks}.
\newblock {\em Discrete Applied Mathematics}, 161:1--12, 2013.

\bibitem{J-Benecke2006}
A.~Benecke.
\newblock {Chromatin code, local non-equilibrium dynamics, and the emergence of
  transcription regulatory programs}.
\newblock {\em The European Physical Journal E}, 19:353--366, 2006.

\bibitem{C-Bridoux2021}
F.~Bridoux, C.~Gaze-Maillot, K.~Perrot, and S.~Sen{\'e}.
\newblock {Complexity of Limit-Cycle Problems in Boolean Networks}.
\newblock In {\em Proceedings of SOFSEM'2021}, volume 12607 of {\em LNCS},
  pages 135--146. Springer, 2021.

\bibitem{C-Bridoux2017}
F.~Bridoux, P.~Guillon, K.~Perrot, S.~Sen{\'{e}}, and G.~Theyssier.
\newblock On the cost of simulating a parallel boolean automata network by a
  block-sequential one.
\newblock In {\em Proceedings of TAMC'2017}, volume 10185 of {\em LNCS}, pages
  112--128, 2017.
\newblock \href {https://arxiv.org/abs/1702.03101} {\path{arXiv:1702.03101}},
  \href {https://doi.org/10.1007/978-3-319-55911-7_9}
  {\path{doi:10.1007/978-3-319-55911-7_9}}.

\bibitem{J-Creignou2019}
N.~Creignou, M.~Kr{\"o}ll, R.~Pichler, S.~Skritek, and H.~Vollmer.
\newblock {A complexity theory for hard enumeration problem}.
\newblock {\em Discrete Applied Mathematics}, 268:191--209, 2019.

\bibitem{C-Datta2010}
S.~Datta, N.~Limaye, P.~Nimbhorkar, T.~Thierauf, and F.~Wagner.
\newblock {Planar Graph Isomorphism is in Log-Space}.
\newblock In {\em Algebraic Methods in Computational Complexity}, volume 9421,
  pages 1--32. Schloss Dagstuhl, 2010.
\newblock \href {https://doi.org/10.4230/DagSemProc.09421.6}
  {\path{doi:10.4230/DagSemProc.09421.6}}.

\bibitem{J-Demongeot2008}
J.~Demongeot, A.~Elena, and S.~Sen{\'e}.
\newblock {Robustness in regulatory networks: a multi-disciplinary approach}.
\newblock {\em Acta Biotheoretica}, 56:27--49, 2008.

\bibitem{J-Demongeot2020}
J.~Demongeot and S.~Sen{\'e}.
\newblock {About block-parallel Boolean networks: a position paper}.
\newblock {\em Natural Computing}, 19:5--13, 2020.

\bibitem{J-Dennunzio2019}
A.~Dennunzio, E.~Formenti, L.~Manzoni, and A.~E. Porreca.
\newblock {Complexity of the dynamics of reaction systems}.
\newblock {\em Information and Computation}, 267:96--109, 2019.

\bibitem{A-Eppstein2023}
David Eppstein.
\newblock {The Complexity of Iterated Reversible Computation}.
\newblock {\em {TheoretiCS}}, 2, 2023.
\newblock \href {https://doi.org/10.46298/theoretics.23.10}
  {\path{doi:10.46298/theoretics.23.10}}.

\bibitem{J-Fierz2019}
B.~Fierz and M.~G. Poirier.
\newblock {Biophysics of chromatin dynamics}.
\newblock {\em Annual Review of Biophysics}, 48:321--345, 2019.

\bibitem{J-Floreen1989}
P.~Flor{\'e}en and P.~Orponen.
\newblock {On the computational complexity of analyzing Hopfield nets}.
\newblock {\em Complex Systems}, 3:577--587, 1989.

\bibitem{C-Gamard2021}
G.~Gamard, P.~Guillon, K.~Perrot, and G.~Theyssier.
\newblock {Rice-like Theorems for automata networks}.
\newblock In {\em Proceedings of STACS'21}, volume 187 of {\em LIPIcs}, pages
  32:1--32:17. Schloss Dagstuhl -- Leibniz-Zentrum f{\"u}r Informatik, 2021.

\bibitem{J-Giacomantonio2010}
C.~E. Giacomantonio and G.~J. Goodhill.
\newblock {A Boolean model of the gene regulatory network underlying mammalian
  cortical area development}.
\newblock {\em PLoS Computational Biology}, 6:e1000936, 2010.

\bibitem{B-Goles1990}
E.~Goles and S.~Mart{\'i}nez.
\newblock {\em {Neural and automata networks: dynamical behavior and
  applications}}, volume~58 of {\em Mathematics and Its Applications}.
\newblock Kluwer Academic Publishers, 1990.

\bibitem{J-Goles2016}
E.~Goles, P.~Montealegre, V.~Salo, and I.~Törmä.
\newblock {PSPACE-completeness of majority automata networks}.
\newblock {\em Theoretical Computer Science}, 609:118--128, 2016.
\newblock \href {https://doi.org/10.1016/j.tcs.2015.09.014}
  {\path{doi:10.1016/j.tcs.2015.09.014}}.

\bibitem{C-Goles2010}
E.~Goles and M.~Noual.
\newblock {Block-sequential update schedules and Boolean automata circuits}.
\newblock In {\em Proceedings of AUTOMATA'2010}, pages 41--50. DMTCS, 2010.

\bibitem{J-Goles2008}
E.~Goles and L.~Salinas.
\newblock {Comparison between parallel and serial dynamics of Boolean
  networks}.
\newblock {\em Theoretical Computer Science}, 396:247--253, 2008.

\bibitem{J-Hansen1992}
J.~C. Hanse and J.~Ausio.
\newblock {Chromatin dynamics and the modulation of genetic activity}.
\newblock {\em Trends in Biochemical Sciences}, 17:187--191, 1992.

\bibitem{J-Hubner2010}
M.~R. H{\"u}bner and D.~L. Spector.
\newblock {Chromatin dynamics}.
\newblock {\em Annual Review of Biophysics}, 39:471--489, 2010.

\bibitem{C-Kari2005}
J.~Kari.
\newblock Reversible cellular automata.
\newblock In {\em Proceedings of DLT'2005}, volume 3572 of {\em LNCS}, pages
  57--68. Springer, 2005.
\newblock \href {https://doi.org/10.1007/11505877_5}
  {\path{doi:10.1007/11505877_5}}.

\bibitem{J-Kauffman1969}
S.~A. Kauffman.
\newblock {Metabolic stability and epigenesis in randomly constructed genetic
  nets}.
\newblock {\em Journal of Theoretical Biology}, 22:437--467, 1969.

\bibitem{J-Kobler1996}
J.~Köbler and S.~Toda.
\newblock On the power of generalized {M}od-classes.
\newblock {\em Mathematical Systems Theory}, 29:33–--46, 1996.
\newblock \href {https://doi.org/10.1007/BF01201812}
  {\path{doi:10.1007/BF01201812}}.

\bibitem{J-Lange2000}
K.-J. Lange, P.~McKenzie, and A.~Tapp.
\newblock {Reversible Space Equals Deterministic Space}.
\newblock {\em Journal of Computer and System Sciences}, 60(2):354--367, 2000.
\newblock \href {https://doi.org/10.1006/jcss.1999.1672}
  {\path{doi:10.1006/jcss.1999.1672}}.

\bibitem{J-McCulloch1943}
W.~S. McCulloch and W.~Pitts.
\newblock {A logical calculus of the ideas immanent in nervous activity}.
\newblock {\em Journal of Mathematical Biophysics}, 5:115--133, 1943.

\bibitem{J-Mendoza1998}
L.~Mendoza and E.~R. Alvarez-Buylla.
\newblock {Dynamics of the genetic regulatory network for \emph{Arabidopsis
  thaliana} flower morphogenesis}.
\newblock {\em Journal of Theoretical Biology}, 193:307--319, 1998.

\bibitem{B-Morita2017}
K.~Morita.
\newblock {\em {Theory of Reversible Computing}}.
\newblock Springer, first edition, 2017.
\newblock \href {https://doi.org/10.1007/978-4-431-56606-9}
  {\path{doi:10.1007/978-4-431-56606-9}}.

\bibitem{ns11}
M.~Noual and S.~Sen{\'e}.
\newblock {Towards a theory of modelling with Boolean automata networks~-~I.
  Theorisation and observations}, 2011.
\newblock arXiv:1111.2077.

\bibitem{BC-Pauleve2022}
L.~Paulev{\'e} and S.~Sen{\'e}.
\newblock {\em {Systems biology modelling and analysis: formal bioinformatics
  methods and tools}}, chapter {Boolean networks and their dynamics: the impact
  of updates}, pages 173--250.
\newblock Wiley, 2022.

\bibitem{HDR-Perrot2022}
K.~Perrot.
\newblock {\em {{\'E}tudes de la complexit{\'e} algorithmique des r{\'e}seaux
  d'automates}}.
\newblock Habilitation thesis, Universit\'e d'Aix-Marseille, 2022.

\bibitem{C-Perrot2024a}
K.~Perrot, S.~Sen{\'e}, and L.~Tapin.
\newblock {Combinatorics of block-parallel automata networks}.
\newblock In {\em Proceedings of SOFSEM'24}, volume 14519 of {\em LNCS}, pages
  442--455. Springer, 2024.

\bibitem{C-Perrot2024b}
K.~Perrot, S.~Sen{\'e}, and L.~Tapin.
\newblock {Complexity of Boolean automata networks under block-parallel update
  modes}.
\newblock In {\em Proceedings of SAND'24}, volume 292 of {\em LIPIcs}, pages
  19:1--19:19. Schloss Daghstul Publishing, 2024.

\bibitem{C-Perrotin2023}
P.~Perrotin and S.~Sen{\'e}.
\newblock {Turning block-sequential automata networks into smaller parallel
  networks with isomorphic limit dynamics}.
\newblock In {\em Proceedings of CiE'23}, volume 13967 of {\em LNCS}, pages
  214--228. Springer, 2023.

\bibitem{B-Robert1986}
F.~Robert.
\newblock {\em {Discrete iterations: a metric study}}, volume~6 of {\em
  Springer Series in Computational Mathematics}.
\newblock Springer, 1986.

\bibitem{J-Sutner1995}
K.~Sutner.
\newblock {On the Computational Complexity of Finite Cellular Automata}.
\newblock {\em Journal of Computer and System Sciences}, 50(1):87--97, 1995.
\newblock \href {https://doi.org/10.1006/jcss.1995.1009}
  {\path{doi:10.1006/jcss.1995.1009}}.

\bibitem{J-Sutner2004}
K.~Sutner.
\newblock The complexity of reversible cellular automata.
\newblock {\em Theoretical Computer Science}, 325(2):317--328, 2004.
\newblock \href {https://doi.org/10.1016/j.tcs.2004.06.011}
  {\path{doi:10.1016/j.tcs.2004.06.011}}.

\bibitem{J-Thomas1973}
R.~Thomas.
\newblock {Boolean formalization of genetic control circuits}.
\newblock {\em Journal of Theoretical Biology}, 42:563--585, 1973.

\bibitem{J-Wooten2019}
D.~J. Wooten, S.~M. Groves, D.~R. Tyson, Q.~Liu, J.~S. Lim, R.~Albert, C.~F.
  Lopez, J.~Sage, and V.~Quaranta.
\newblock {Systems-level network modeling of small cell lung cancer subtypes
  identifies master regulators and destabilizers}.
\newblock {\em PLoS Computational Biology}, 15:e1007343, 2019.

\end{thebibliography}

%%%%%%%%%%%%%%%%%%%%%%%%%%%%%%%%

\newpage
\appendix

%%%%%%%%%%%%%%%%%%%%%%%%%%%%%%%%
\section{Omitted proofs}
\label{a:proofs}

\begin{proof}[Continuation of the proof of Theorem~\ref{theorem:BPn0}]
  Formula~\ref{BPn0_eq1} is a sum for each partition of $n$ (sum 
  on $i$), of all the ways to fill all the matrices ($n!$) up to permutation
  within each column ($m(i,j)!$ for each of the $j$ columns of $M_j$).%\medskip

  Formula~\ref{BPn0_eq2} is a sum for each partition of $n$ (sum on $i$),
  of the product for each column of the matrices (products on $j$ and $\ell$),
  of the choice of elements (among the remaining ones) to fill the column
  (regardless of their order within the column).%\medskip

  Formula~\ref{BPn0_eq3} is a sum for each partition of $n$ (sum on $i$),
  of the product for each matrix (product on $j$), of the choice of elements
  (among the remaining ones) to fill this matrix, multiplied by the number of 
  ways to fill the columns of the matrix (product on $\ell$) with these elements
  (regardless of their order within each column).%\medskip

  The equality between Formulas~\ref{BPn0_eq1} and~\ref{BPn0_eq2} is obtained 
  by developing the binomial coefficients as follows: 
  ${\binom{x}{y}} = \frac{x!}{y!\cdot(x-y)!}$, 
  and by observing that the products of $\frac{x!}{(x-y)!}$ telescope.
  Indeed, denoting 
  $a(j,\ell) = (n - \sum_{k = 1}^{j-1} k \cdot m(i,k) - \ell \cdot m(i,j))!$,
  we have
  \[
    \prod_{j = 1}^{d(i)} 
    	\prod_{\ell = 1}^{j}
				\frac
					{(n - \sum_{k = 1}^{j - 1} k \cdot m(i,k) - (\ell - 1) \cdot m(i,j))!}
					{(n - \sum_{k = 1}^{j - 1} k \cdot m(i,k) - \ell \cdot m(i,j))!}
    \;=\;
    \prod_{j = 1}^{d(i)} 
    	\prod_{\ell = 1}^{j}
    		\frac
					{a(j,\ell-1)}
					{a(j,\ell)}
    \;=\;
    \frac{n!}{0!}
    \;=\;
    n!
  \]
  because $a(1,0) = n!$, then $a(1,j) = a(2,0)$, $a(2,j) = a(3,0)$, ..., until 
  $a(d(i),j) = 0!$.\medskip

  The equality between Formulas~\ref{BPn0_eq2} and~\ref{BPn0_eq3} is obtained by 
  repeated uses of the identity
  ${\binom{x}{z}}{\binom{x-z}{y}} = {\binom{x}{z+y}}{\binom{z+y}{y}}$, 
  which gives by induction on $j$:
  \begin{align}
  	\label{eq:binom23}
    \prod_{\ell = 1}^{j} 
    	\binom{x - (\ell - 1) \cdot y}{y}
    \;=\;
    \binom{x}{j \cdot y} \cdot \prod_{\ell = 1}^{j} 
    	\binom{(j - \ell + 1) \cdot y}{y}\text{.}
  \end{align}
  Indeed, $j=1$ is trivial and, using the induction hypothesis on $j$ then the 
  identity we get:
  \begin{align*}
    \prod_{\ell = 1}^{j+1} 
    	\binom{x - (\ell - 1) \cdot y}{y}
    & \;=\; 
    	\binom{x - j \cdot y}{y} \cdot \prod_{\ell = 1}^{j} 
	    	\binom{x - (\ell - 1) \cdot y}{y}\\
    & \;=\; 
    	\binom{x - j \cdot y}{y} \cdot \binom{x}{j \cdot y} 
    	\cdot \prod_{\ell = 1}^{j} 
    		\binom{(j - \ell + 1) \cdot y}{y}\\
    & \;=\;
    	\binom{x}{(j + 1) \cdot y} \cdot \binom{(j + 1) \cdot y}{y} 
    	\cdot \prod_{\ell = 1}^{j} 
    		\binom{(j - \ell + 1) \cdot y}{y}\\
    & \;=\;
	    \binom{x}{(j + 1) \cdot y} \cdot 
	    \prod_{\ell = 0}^{j} 
	    	\binom{(j - \ell + 1) \cdot y}{y}\\
    & \;=\;
    	\binom{x}{(j + 1) \cdot y} \cdot \prod_{\ell = 1}^{j + 1} 
    		\binom{(j + 1 - \ell + 1) \cdot y}{y}\text{.}
  \end{align*}
  As a result, Formula~\ref{BPn0_eq3} is obtained from Formula~\ref{BPn0_eq2} 
  by applying Equation~\ref{eq:binom23} for each $j$ with 
  $x = n - \sum_{k = 1}^{j-1} k \cdot m(i,k)$ and $y = m(i,j)$.\medskip

\end{proof}

\begin{proof}[Proof of Lemma~\ref{l:egf}]
	We will start from the exponential generating function by finding the
	coefficient of $x^n$ and proving that it is equal to $\frac{|\BPn^0|}{n!}$, 
	and thus that the associated sequence is $(|\BPn^0|)_{n\in\N}$.

	\begin{multline*}
		\prod_{j \geq 1} 
			\sum_{k \geq 0} 
				\left( 
					\frac{x^k}{k!} 
				\right)^j
		\;=\;
			\left( 
				\sum_{k \geq 0} 
					\frac{x^k}{k!} 
			\right) 
			\times 
			\left( 
				\sum_{k \geq 0} 
					\frac{x^{2k}}{(k!)^2} 
			\right) 
			\times 
			\left( 
				\sum_{k \geq 0} 
					\frac{x^{3k}}{(k!)^3} 
			\right) 
			\times  
			\cdots \\
		\;=\; 
			\underbrace{
				\left( 1 + x + \frac{x^2}{2!} + \cdots \right)
			}_{j = 1} 
			\times
			\underbrace{
				\left( 1 + x^2 + \frac{x^4}{(2!)^2} + \cdots \right)
			}_{j = 2}
			\times
			\underbrace{
				\left( 1 + x^3 + \frac{x^6}{(2!)^3} + \cdots \right)
			}_{j = 3}
			\times
			\cdots
			\text{.}
	\end{multline*}
%\begin{align*}
%	\prod_{j \geq 1} \sum_{k \geq 0} \left(\frac{x^k}{k!} \right)^j 
%	& \;=\;
%		\left( \sum_{k \geq 0} \frac{x^k}{k!} \right) \times 
%		\left( \sum_{k \geq 0} \frac{x^{2k}}{(k!)^2} \right) \times 
%		\left( \sum_{k \geq 0} \frac{x^{3k}}{(k!)^3} \right) \times  
%		\cdots\\
%	& \;=\; 
%		\underbrace{\left( 1 + x + \frac{x^2}{2!} + \cdots \right)}_{j = 1} 
%		\times
%		\underbrace{\left( 1 + x^2 + \frac{x^4}{(2!)^2} + \cdots \right)}_{j = 2}
%		\times
%		\underbrace{\left( 1 + x^3 + \frac{x^6}{(2!)^3} + \cdots \right)}_{j = 3}
%		\times
%		\cdots
%\end{align*}
	Each term of the distributed sum is obtained by associating a $k \in \N$ to 
	each $j \in \N_+$, and by doing the product of the $\frac{1}{(k!)^j} \cdot 
	x^{jk}$.
	Thus, if $\N^{\N_+}$ is the set of maps from $\N_+$ to  $\N$, we have:
	\[
		\prod_{j \geq 1} 
			\sum_{k \geq 0} \left( 
				\frac{x^k}{k!} 
			\right)^j 
		\;=\;
			\sum_{m \in \N^{\N_+}} 
				\left( 
					\prod_{j \geq 1} 
						\frac{1}{(m(j)!)^j} 
				\right)
				\cdot 
				x^{\sum_{j \geq 1} j \cdot m(j)}\text{.}
	\]
	From here, to get the coefficient of $x^n$, we need to do the sum only on the
	maps $m$ such that $\sum_{j\geq1}j\cdot m(j) = n$, which just so happen to be
	the partitions of $n$, with $m(j)$ being the multiplicity of $j$ in the
	partition.
	Thus, the coefficient of $x^n$ is
	\[
  	\sum_{i = 1}^{p(n)} 
  		\prod_{j \geq 1} 
  			\frac{1}{(m(i,j)!)^j} 
	  \;=\; 
	  	\sum_{i = 1}^{p(n)} 
	  		\frac{1}{\prod_{j \geq 1}^{d(i)} (m(i,j)!)^j} 
	  \;=\; 
	  	\frac{|\BPn^0|}{n!}\text{.}
	\]
\end{proof}

\begin{proof}[Proof of correction of Algorithm~\ref{algo:BPiso}]
  We first argue that Algorithm~\ref{algo:BPiso} enumerates the correct number
  of block-parallel update modes, and then that any pair $\mu, \mu'$
  enumerated is such that $\mu \not \equiv_\star \mu'$.

	For ease of reading in the rest of this proof, we will denote $a[j]$ as $a_j$
	and $m(i, j)$ as $m_{ij}$.
	From the placement of $min_j$ described above (forced to be within the first 
	$a_j$ columns of matrix $M_j$), the difference with the previous algorithm
	is that, instead of having $\prod_{\ell = 1}^{j} \binom{(j - \ell + 1) \cdot
	m_{ij}}{m_{ij}}$ ways of filling matrix $M_j$, we only have the following
	number of ways (recall that $M_j$ has $j$ columns and $m_{ij}$ rows):
	\[
		\sum_{k = 1}^{a_j} 
			\left( 
				\prod_{\ell = 1}^{k - 1} 
					\binom{(j - \ell + 1) \cdot m_{ij} - 1}{m_{ij}} 
			\right) 
			\cdot 
			\binom{(j - k + 1) \cdot m_{ij} - 1}{m_{ij} - 1} 
			\cdot 
			\left(
				\prod_{\ell = k + 1}^{j}
					\binom{(j - \ell + 1) \cdot m_{ij}}{m_{ij}}
			\right)\text{.}
	\]
	Indeed, the formula above sums, for each choice of a column $k$ from $1$ to 
	$a_j$ where $min_j$ will be placed, the number of ways to place some elements 
	within columns $1$ to $k-1$ (first product on $\ell$), times the number of 
	ways to choose some elements that will accompany $min_j$ within column $k$ 
	(middle binomial coefficient), times the number of ways to place some other 
	elements within the remaining columns $k+1$ to $j$ (second product on $\ell$).
	Now, we have 
	$\binom{(j - k + 1) \cdot m_{ij}-1}{m_{ij} - 1} 
		= 
			\frac{m_{ij}}{(j - k + 1) \cdot m_{ij}} 
			\cdot 
			\binom{(j - k + 1) \cdot m_{ij}}{m_{ij}} 
		= 
			\frac{1}{(j - k + 1)} \cdot \binom{(j - k + 1) 
			\cdot 
			m_{ij}}{m_{ij}}$. 
	We also have
	$\binom{(j - \ell + 1) \cdot m_{ij} - 1}{m_{ij}}
		=
			\frac{j - l}{j - (l - 1)} 
			\cdot
			\binom{(j - \ell + 1) \cdot m_{ij}}{m_{ij}}$.
	This means that the sum of the possible ways to choose the content
	of matrix $M_j$ can be rewritten as follows:
	\begin{multline*}
		\sum_{k = 1}^{a_j}
			\left(
				\frac{1}{(j - k + 1)} \cdot \prod_{\ell = 1}^{k - 1}
					\frac{j - \ell}{j - (\ell - 1)} \cdot \prod_{\ell = 1}^j 
						\binom{(j - \ell + 1) \cdot m_{ij}}{m_{ij}}
			\right)\\ 
		\begin{split}
			\;=\; & 
				\sum_{k = 1}^{a_j}
					\left(
						\frac{1}{(j - k + 1)} \cdot \frac{j - k + 1}{j} \cdot 
						\prod_{\ell = 1}^j
							\binom{(j - \ell + 1) \cdot m_{ij}}{m_{ij}}
					\right)\\
			\;=\; & 
				\sum_{k = 1}^{a_j} 
					\left(
						\frac{1}{j} \cdot \prod_{\ell = 1}^j 
							\binom{(j - \ell + 1) \cdot m_{ij}}{m_{ij}}
					\right)\\
		  \;=\; & 
		  	\frac{a_j}{j} \cdot \prod_{\ell = 1}^j
		  		\binom{(j - \ell + 1) \cdot m_{ij}}{m_{ij}}\text{.}
		\end{split}
	\end{multline*}
	Hence in total, the algorithm enumerates the following number of update modes:
	\[
		\sum_{i = 1}^{p(n)}
			\prod_{j = 1}^{d(i)}
				\left(
					\binom{n - \sum_{k = 1}^{j - 1} k \cdot m_{ik}}{j \cdot m_{ij}}
					\cdot
					\frac{a_j}{j} \cdot \prod_{\ell = 1}^{j}
						\binom{(j - \ell + 1) \cdot m_{ij}}{m_{ij}}
    		\right)\text{.}
  \]
	In order to prove that this number is equal to $|\BPn^\star|$, we need to 
	prove that $\prod_{j = 1}^{d(i)} \frac{a_j}{j} = \frac{1}{\lcm(i)}$.
	Denoting $L(j) = \lcm(\{k \in \entiers{j, d(i)} \mid m_{ik} > 0\})$, we prove 
	by induction that at the end of each step of the \textbf{for} loop from lines 
	5-10, we have:
	\[
		\prod_{k = j}^{d(i)} 
			\frac{a_k}{k} 
		= 
			\frac{1}{L(j)}\text{,}
	\] 
	and the claim follows (when $j = 1$, we get $L(j) = \lcm(i)$).
	At the first step, $j = d(i)$, and 
	\[
		\frac{a_{d(i)}}{d(i)} 
		= 
			\frac{\gcd(\{d(i), 1\})}{d(i)} 
		= 
			\frac{1}{L(j)}\text{.}
	\]
	We assume as induction hypothesis that for a given $j$, we have
	$\prod_{k = j}^{d(i)}\frac{a_k}{k} = \frac{1}{L(j)}$.
	There are two possible cases for $j - 1$:
	\begin{itemize}[nosep]
	\item If $m_{i(j-1)} = 0$, then 
		\[
			a_{j-1} = j-1 \text{ and }
			\prod_{k = j - 1}^{d(i)}
				\frac{a_k}{k} 
			= 
				\frac{j - 1}{(j - 1)\cdot L(j)} 
			= 
				\frac{1}{L(j-1)}\text{.}
		\]
	\item Otherwise, 
		\[
			\prod_{k = j - 1}^{d(i)}
				\frac{a_k}{k} 
			= 
				\frac{\gcd(\{j - 1, L(j)\})}{(j - 1) \cdot L(j)}\text{.}
		\]
		And since $\frac{a \cdot b}{\gcd(\{a, b\})} = \lcm(\{a, b\})$,
		we have
		\[
			\prod_{k = j - 1}^{d(i)}
				\frac{a_k}{k} 
			= 
				\frac{1}{\lcm(\{j-1, L(j)\})} 
			=
		    \frac{1}{L(j-1)}\text{.}
		\]
	\end{itemize}
	We conclude that at the end of the loop, we have $\prod_{j = 1}^{d(i)}
		\frac{a_j}{j} = \frac{1}{\lcm(i)}$, and thus that the algorithm enumerates 
                the following number of update modes (cf.~Remark~\ref{remark:BPnstar}):
	\[
		\sum_{i = 1}^{p(n)} 
			\prod_{j = 1}^{d(i)}
				\left(
					\binom{n - \sum_{k = 1}^{j-1} k \cdot m_{ik}}{j \cdot m_{ij}}
					\cdot
					\prod_{\ell = 1}^{j}
						\binom{(j - \ell + 1) \cdot m_{ij}}{m_{ij}}
				\right) 
			\cdot 
			\frac{1}{\lcm(i)} 
		= 
			|\BPn^\star|\text{.}
	\]\medskip

We now need to prove that the algorithm does not enumerate two equivalent
update modes (in the sense of $\equiv_\star$).
Algorithm~\ref{algo:BPiso} is heavily based on the algorithm from the previous section, in
such a way that, for a given input, the output of Algorithm~\ref{algo:BPiso} 
will be a subset of that of the aforementioned algorithm.
Said algorithm enumerates $\BPn^0$, which implies that every update
mode enumerated by it has a different image by $\varphi$. 
This means that every block-parallel update mode enumerated by 
Algorithm~\ref{algo:BPiso} also has a different image by $\varphi$, and that the 
algorithm does not enumerate two equivalent update modes with a shift of $0$.

We now prove by contradiction that the algorithm does not enumerate two 
equivalent update modes with a non-zero shift. 
Let $\mu, \mu' \in \BPn$ be two update modes, both enumerated by 
Algorithm~\ref{algo:BPiso}, such that $\mu \equiv_\star \mu'$ with a non-zero 
shift. 
Then, there is $k \in \entiers[1]{|\varphi(\mu')| - 1}$ such that 
$\varphi(\mu) = \sigma^k(\varphi(\mu'))$, and $\mu, \mu'$ are both generated
from the same partition $i$, with $|\varphi(\mu)| = |\varphi(\mu')| = \lcm(i)$.
Moreover, for each $j \in \entiers{d(i)}$ the matrix $M_j$ must contain
the same elements $A$ in both $\mu$ and $\mu'$ (so that they are repeated every 
$j$ blocks in both $\varphi(\mu)$ and $\varphi(\mu')$), hence in particular 
$min_j$ is the same in both enumerations.

We will prove by induction that every $j \in \entiers{d(i)}$ such that 
$m(i, j) > 0$ divides $k$, and therefore $k = 0$ (equivalently $k = \lcm(i)$), 
leading to a contradiction.
For the base case $j = d(i)$, we have $a_{d(i)} = 1$ (first iteration of the 
\textbf{for} loop lines 5-10), hence the call of \texttt{EnumBlockIsoAux} for 
any set $A$ passes the condition of line 35 and $min_{d(i)}$ is immediately 
chosen to belong to $C_1$ (the first column of $M_{d(i)}$).
When converted to block-sequential update modes, it means that $min_{d(i)}$
appears in all blocks indexed by $d(i) \cdot t$ with $t \in \N$, hence $k$ must 
be a multiple of $d(i)$ so that $\sigma^k$ maps blocks containing $min_{d(i)}$
to blocks containing $min_{d(i)}$.

As induction hypothesis, assume that for a given $j$, every $\ell \in 
\llbracket j, d(i)\rrbracket$ such that $m_{i\ell} > 0$ divides $k$. 
We will prove that $j'$, the biggest number in the partition $i$ 
(\emph{i.e.}~with $m_{ij'} > 0$) that is smaller than $j$, also divides $k$.
In the matrix $M_{j'}$, the minimum $min_{j'}$ is forced to appear within the 
$a_{j'}$ first columns.
This means that block indexes where it appears in both $\varphi(\mu)$ and 
$\varphi(\mu')$ can be written respectively as $j'\cdot t + b$ and 
$j'\cdot t + b'$ respectively, with $t \in \N$ and $b , b'\in 
\llbracket 1, a_{j'}\rrbracket$.
As a consequence, an automaton from $M_{j'}$ that is at the position $b$ in
$\varphi(\mu)$ is at a position of the form $j' \cdot t + b'$ in 
$\varphi(\mu')$. 
It follows that $b + k = j'\cdot t + b'$, which can be rewritten as 
$k = t \cdot j' + b' - b = t \cdot j' + d$, with $t \in \N$ and 
$d = b'- b \in \llbracket - a_{j'} + 1, a_{j'} - 1 \rrbracket$.
Moreover, we know by induction hypothesis that every number in the partition $i$
that is greater than $j'$ divides $k$, making $k$ a common multiple of these
numbers. 
We deduce that their lowest common multiple also divides $k$. 
Given that $a_{j'}$ is the $\gcd$ of $j'$ and said $\lcm$ (lines 7-8), it means 
that $a_{j'}$ divides both $j'$ and $k$, which implies that it also divides $d$.
Since $d$ is in $\llbracket -a_{j'} + 1, a_{j'} - 1 \rrbracket$, we have $d = 0$
and thus, $j'$ divides $k$. This concludes the induction.

If every number of the partition $i$ divides $k$ and $k \in \entiers{\lcm(i)}$,
then $k=0$, leading to a contradiction. 
This concludes the proof of correctness.
\end{proof}

\begin{proof}[Proof of Lemma~\ref{lem:primes}]
  By the prime number theorem, there are approximately $\frac{N}{\ln(N)}$ primes lower than $N$.
  As a consequence, distinct prime integers $p_1,p_2,\dots,p_{k_n}$ with $k_n=\lfloor\frac{n^2}{\ln(n^2)}\rfloor$
  can be computed in time $\O(n^2)$ using Atkin sieve algorithm.
  Since $2 \leq p_i < n^2$, we have $2^{k_n} \leq \prod_{i=1}^{k_n} p_i < n^{2k_n}$.
  %there are approximately $\frac{N}{\ln(N)}$
  %primes lower than a given $N\in\N$, which can be computed in time $\O(N)$
  %using Atkin sieve algorithm.
  %Consider such primes $p_1,p_2,\dots,p_k$ obtained for some $N=n^2$.
  It holds that $2^{k_n} = 2^{\lfloor\frac{n^2}{2\ln(n)}\rfloor} > 2^n$, and
  $n^{2k_n} \leq n^{\frac{n^2}{\ln(n)}}$ with
  \[
    %\log_2\left(N^{\frac{N}{\ln(N)}}\right) =
    %\log_2\left(n^\frac{2n^2}{\ln(n^2)}\right) =
    \log_2\left(n^\frac{n^2}{\ln(n)}\right) =
    \frac{\frac{n^2}{\ln(n)}}{\log_n(2)} =
    \frac{n^2}{\ln(2)}
  \]
  meaning that $n^{2k_n} \leq 2^{\frac{n^2}{\ln(2)}} < 2^{2n^2}$.
\end{proof}

\noindent
\decisionpb{Block-parallel preimage}{BP-Preimage}
{$(f_i:\B^n\to\B)_{i\in\entiers{n}}$ as circuits, $\mu\in\BPn$, $y\in\B^n$.}
{does $\exists x\in\B^n:\mmjoblock{f}{\mu}(x)=y$?} %$|\mmjoblock{f}{\mu}^{-1}(y)|\neq\emptyset$

\begin{theorem}\label{thm:preimage}
  {\normalfont \textbf{BP-Preimage}} is $\PSPACE$-complete.
\end{theorem}

The difficulty in this reduction is that we need to take into account the image of every configuration $x$.
We modify the preceding construction by setting automata $D$ to $\tilde{x}$
when the counter $B$ encodes $0$.\\[.5em]

\begin{proof}
  The algorithm for \textbf{BP-Preimage} computes the image of each configuration
  (enumerated in polynomial space with a simple counter) using the same procedure as \textbf{BP-Step}
  (Theorem~\ref{thm:image}),
  and decides whether there is some $x$ such that $\mmjoblock{f}{\mu}(x)=y$.

  \medskip

  Given an instance $C:\B^n\to\B^n$, $\tilde{x}\in\B^n$, $i\in\entiers{n}$ of \textbf{Iter-CVP},
  we construct the same block-parallel update schedule $\mu'$ as in the proof of Theorem~\ref{thm:image},
  and modify the local functions of automata $D$ and $R$ as follows:
  \[
    f_D(x)=
    \begin{cases}
      C(\tilde{x})&\text{if }x_B=0\\
      C(x_D)&\text{if }0<x_B<\ell-1\\
      \0^n&\text{otherwise}
    \end{cases}
    \qquad
    f_R(x)=
    \begin{cases}
      \tilde{x}_i&\text{if }x_B=0\\
      x_R\vee x_{q_{k_n}+\ell'+i}&\text{otherwise}
    \end{cases}
  \]
  The purpose is that $D$ iterates the circuit from $\tilde{x}$ when the counter is initialized to $0$,
  and that $R$ records whether the $i$-th bit of $D$ has been in state $\1$
  (including the initial substep).
  We set $y=\0^{q_{k_n}}\0^\ell\0^n\1$.

  If the \textbf{Iter-CVP} instance is positive,
  then we have $\mmjoblock{f}{\mu'}(\0^{q_{k_n}}\0^\ell\0^n\0)=y$
  (automata $B$ go back to $\0^{q_{k_n}}$,
  automata $D$ iterate circuit $C$ from $\tilde{x}$ and end in state $\0^n$,
  and automaton $R$ has recorded that the $i$-th bit of $D$ has been to state $\1$).

  Conversely, if there is a configuration $x$ such that $\mmjoblock{f}{\mu'}(x)=y$,
  then the automata from the counter $B$ must have started in state $x_B=\0^{q_{k_n}}$,
  because of the increment modulo $\ell$ which is the number of substeps.
  We deduce that $D$ iterate circuit $C$ for the whole orbit of $\tilde{x}$ and end in state $\0^n$,
  and that automaton $R$ records the answer to the \textbf{Iter-CVP} instance.
  Since it it ends in state $y_R=\1$ by our assumption that $\mmjoblock{f}{\mu'}(x)=y$,
  we conclude that it is positive.
\end{proof}

\noindent
\decisionpb{Block-parallel reachability}{BP-Reachability}
{$(f_i:\B^n\to\B)_{i\in\entiers{n}}$ as circuits, $\mu\in\BPn$, $x,y\in\B^n$.}
{does $\exists t\in\N:\mmjoblock{f}{\mu}^t(x)=y$?}

\begin{theorem}
  \label{thm:reach}
  {\normalfont \textbf{BP-Reachability}} is $\PSPACE$-complete.
\end{theorem}

\begin{proof}
  The problem belongs to $\PSPACE$, because is can naively be solved by simulating
  the dynamics of $\mmjoblock{f}{\mu}$ starting from configuration $x$, for $2^n$ time steps.

  \medskip

  Reachability problems in cellular automata and related models are known to be $\PSPACE$-complete
  on finite configurations~\cite{J-Sutner1995}.
  We reduce from the reachability problem for reaction systems,
  which can be seen as a particular case of Boolean automata networks,
  and is also known to be $\PSPACE$-complete~\cite{J-Dennunzio2019}.
  Given a reaction system $(S,A)$ where $S$ is a finite set of entities,
  and $A$ is a set of reactions of the form $(R,I,P)$
  where $R$ are the reactants, $I$ the inhibitors and $P$ the products,
  we construct the BAN of size $n=|S|$ with local functions:
  \[
    \forall i\in\entiers{n}: f_i(x)=\bigvee_{\substack{(R,I,P)\in A\\\text{such that }i\in P}}
    \left( \bigwedge_{j\in R}x_j \wedge \bigwedge_{k\in I}\neg x_k \right).
  \]
  A configuration $x\in\B^n$ of the BAN corresponds to a state of the reaction system with
  each automaton indicating the presence or absence of its corresponding entity.
  The parallel evolution of $f$ (under $\mu_\texttt{par}$) is in direct correspondance with the evolution
  of the reaction system.
\end{proof}

\begin{proof}[Proof of theorem~\ref{thm:idmodp}]
  Given a formula $\psi$ on $n$ variables, $m$ and $i$ in unary, 
  we apply Lemma~\ref{lem:gn}
  to construct $g_n,\mu_n$ on automata set $P=\entiers{q_{k_n}}$.
  Automata from $P$ have identity local functions,
  and the number of substeps is $\ell=|\varphi(\mu_n)|>2^n$.
  Let $p_i$ be the $i$-th prime number.
  We add:
  \begin{itemize}[nosep]
    \item $\ell'=\lceil\log_2(\ell)\rceil$ automata numbered $B=\{q_{k_n},\dots,q_{k_n}+\ell'-1\}$,
      implementing a $\ell'$ bits binary counter that increments
      modulo $\ell$ at each substep, except for configurations with a counter
      greater of equal to $\ell$ which are left unchanged.
    \item $\ell''=\lceil\log_2(p_i)\rceil$ automata numbered $R=\{q_{k_n}+\ell',\dots,q_{k_n}+\ell'+\ell''-1\}$,
      whose local functions are:
      \[
        f_R(x)=\begin{cases}
          x_R - m + 1 \mod p_i &\text{if } x_B=0 \text{ and } x_B \text{ satisfies } \psi\\
          x_R - m \mod p_i &\text{if } x_B=0 \text{ and } x_B \text{ does not satisfy } \psi\\
          x_R + 1 \mod p_i &\text{if } 0<x_B<2^n \text{ and } x_B \text{ satisfies } \psi\\
          x_R &\text{otherwise.}
        \end{cases}
      \]
  \end{itemize}
  We also add singletons to $\mu_n$ for each of these additional automata,
  with $\mu'=\mu_n\cup\bigcup_{j\in B\cup R}\{(j)\}$.
  The resulting dynamics of $\mmjoblock{f}{\mu'}$ proceeds as follows.

  Configurations $x$ such that $x_B\geq\ell$ verify $\mmjoblock{f}{\mu'}(x)=x$,
  because all local functions are identities in this case.
  For configurations $x$ such that $x_B<\ell$,
  during the dynamics of substeps from $x$ to $\mmjoblock{f}{\mu'}(x)$,
  the counter $x_B$ takes exactly once the values from $0$ to $\ell-1$,
  with $\mmjoblock{f}{\mu'}(x)_B=x_B$ (it goes back to its initial value).
  Meanwhile, at each substep with $x_B<2^n$,
  the record of automata $R$ is incremented if and only if $x_B$ satisfies $\psi$,
  with a substraction of $m$  when $x_B=0$.
  Since $\ell>2^n$ each valuation of $\psi$ is checked exactly once,
  and $x_R$ gets added the number of models of $\psi$ minus $m$, modulo $p_i$
  (when $2^n\leq x_B<\ell$ automata $R$ are left unchanged).
  Consequently, we have $\mmjoblock{f}{\mu'}(x)_R=x_R$ if and only if it has been incremented
  $m$ times modulo $p_i$, i.e., $f,\mu'$
  is a positive instance of \textbf{BP-Identity} if and only if
  $\psi$, $m$, $i$ is a positive instance of \textbf{Mod-SAT}
  (the number of models of $\psi$ is congruent to $k$ modulo $p_i$).
\end{proof}

\begin{proof}[Proof of corollary~\ref{cor:fpspacecomplete}]
  For a fixed reversible cellular automaton (of any dimension),
  given a configuration of size $n$ and a time $t$,
  one can compute in polynomial time a block-parallel update schedule $\mu$
  and circuits for the local functions of a Boolean automata network
  of large enough size (to encode the CA's state space in binary),
  such that:
  \begin{itemize}[nosep]
    \item $|\varphi(\mu)|>t$ (by Lemma~\ref{lem:gn};
      these automata are left aside with identity local functions),
    \item one substep of $\mmjoblock{f}{\mu}$ simulates
      one step of the CA; and
    \item $\mmjoblock{f}{\mu}$ is bijective (because the CA is reversible, padding with identity).
      %KP: there is some padding when |S|^n is not a power of two,
      %    we get a non-surjective encoding to B^n.
  \end{itemize}
  This gives a functional Turing many-one reduction from Theorem~\ref{thm:rca}.
\end{proof}

\end{document}